\documentclass[english]{article}

\usepackage{eurosym}
\usepackage{xspace}
\usepackage{ifthen}
\usepackage{titlesec}
\titlespacing*{\section}{0ex}{1ex}{1.5ex}
\titlespacing*{\paragraph}{0ex}{1ex}{1.3ex}
\titlespacing*{\optionalsubsection}{0ex}{1ex}{1ex}
\usepackage{enumitem}
\setitemize{itemsep=0pt,topsep=0pt,parsep=0pt,partopsep=0pt}
\setenumerate{itemsep=0pt,topsep=0pt,parsep=0pt,partopsep=0pt}
\setdescription{itemsep=0pt,topsep=0pt,parsep=0pt,partopsep=0pt}

\usepackage[english]{babel}  
\usepackage{cite}
\usepackage{amsmath,amssymb,amsfonts}
\usepackage{algorithmic}
\usepackage{graphicx}
\usepackage{bussproofs}
\EnableBpAbbreviations
\usepackage{hyperref}
\usepackage{tikz,pgfplots}
\pgfplotsset{compat=newest}
\usepackage{tikz-cd}
\usepackage{textcomp}
\usepackage{xcolor}
\usepackage{amsthm}
\usepackage{caption}
\usepackage{subcaption}
\usetikzlibrary{cd}
\usepackage{mathabx}
\usepackage{cmll}
\usepackage{stmaryrd}
\usepackage{wrapfig}
\usepackage{relsize}

\usetikzlibrary{intersections,through}

\newdimen\proofrulebreadth \proofrulebreadth=.05em
\newdimen\proofdotseparation \proofdotseparation=1.25ex
\newdimen\proofrulebaseline \proofrulebaseline=2ex
\newcount\proofdotnumber \proofdotnumber=3
\let\then\relax
\def\hfi{\hskip0pt plus.0001fil}
\mathchardef\squigto="3A3B
\newif\ifinsideprooftree\insideprooftreefalse
\newif\ifonleftofproofrule\onleftofproofrulefalse
\newif\ifproofdots\proofdotsfalse
\newif\ifdoubleproof\doubleprooffalse
\let\wereinproofbit\relax
\newdimen\shortenproofleft
\newdimen\shortenproofright
\newdimen\proofbelowshift
\newbox\proofabove
\newbox\proofbelow
\newbox\proofrulename
\def\shiftproofbelow{\let\next\relax\afterassignment\setshiftproofbelow\dimen0 }
\def\shiftproofbelowneg{\def\next{\multiply\dimen0 by-1 }%
\afterassignment\setshiftproofbelow\dimen0 }
\def\setshiftproofbelow{\next\proofbelowshift=\dimen0 }
\def\setproofrulebreadth{\proofrulebreadth}

\def\prooftree{
\ifnum  \lastpenalty=1
\then   \unpenalty
\else   \onleftofproofrulefalse
\fi
\ifonleftofproofrule
\else   \ifinsideprooftree
        \then   \hskip.5em plus1fil
        \fi
\fi
\bgroup
\setbox\proofbelow=\hbox{}\setbox\proofrulename=\hbox{}%
\let\justifies\proofover\let\leadsto\proofoverdots\let\Justifies\proofoverdbl
\let\using\proofusing\let\[\prooftree
\ifinsideprooftree\let\]\endprooftree\fi
\proofdotsfalse\doubleprooffalse
\let\thickness\setproofrulebreadth
\let\shiftright\shiftproofbelow \let\shift\shiftproofbelow
\let\shiftleft\shiftproofbelowneg
\let\ifwasinsideprooftree\ifinsideprooftree
\insideprooftreetrue
\setbox\proofabove=\hbox\bgroup$\displaystyle 
\let\wereinproofbit\prooftree
\shortenproofleft=0pt \shortenproofright=0pt \proofbelowshift=0pt
\onleftofproofruletrue\penalty1
}

\def\eproofbit{
\ifx    \wereinproofbit\prooftree
\then   \ifcase \lastpenalty
        \then   \shortenproofright=0pt  
        \or     \unpenalty\hfil         
        \or     \unpenalty\unskip       
        \else   \shortenproofright=0pt  
        \fi
\fi
\global\dimen0=\shortenproofleft
\global\dimen1=\shortenproofright
\global\dimen2=\proofrulebreadth
\global\dimen3=\proofbelowshift
\global\dimen4=\proofdotseparation
\global\count255=\proofdotnumber
$\egroup  
\shortenproofleft=\dimen0
\shortenproofright=\dimen1
\proofrulebreadth=\dimen2
\proofbelowshift=\dimen3
\proofdotseparation=\dimen4
\proofdotnumber=\count255
}

\def\proofover{
\eproofbit 
\setbox\proofbelow=\hbox\bgroup 
\let\wereinproofbit\proofover
$\displaystyle
}%
\def\proofoverdbl{
\eproofbit 
\doubleprooftrue
\setbox\proofbelow=\hbox\bgroup 
\let\wereinproofbit\proofoverdbl
$\displaystyle
}%
\def\proofoverdots{
\eproofbit 
\proofdotstrue
\setbox\proofbelow=\hbox\bgroup 
\let\wereinproofbit\proofoverdots
$\displaystyle
}%
\def\proofusing{
\eproofbit 
\setbox\proofrulename=\hbox\bgroup 
\let\wereinproofbit\proofusing
\kern0.3em$
}

\def\endprooftree{
\eproofbit 
  \dimen5 =0pt
\dimen0=\wd\proofabove \advance\dimen0-\shortenproofleft
\advance\dimen0-\shortenproofright
\dimen1=.5\dimen0 \advance\dimen1-.5\wd\proofbelow
\dimen4=\dimen1
\advance\dimen1\proofbelowshift \advance\dimen4-\proofbelowshift
\ifdim  \dimen1<0pt
\then   \advance\shortenproofleft\dimen1
        \advance\dimen0-\dimen1
        \dimen1=0pt
        \ifdim  \shortenproofleft<0pt
        \then   \setbox\proofabove=\hbox{%
                        \kern-\shortenproofleft\unhbox\proofabove}%
                \shortenproofleft=0pt
        \fi
\fi
\ifdim  \dimen4<0pt
\then   \advance\shortenproofright\dimen4
        \advance\dimen0-\dimen4
        \dimen4=0pt
\fi
\ifdim  \shortenproofright<\wd\proofrulename
\then   \shortenproofright=\wd\proofrulename
\fi
\dimen2=\shortenproofleft \advance\dimen2 by\dimen1
\dimen3=\shortenproofright\advance\dimen3 by\dimen4
\ifproofdots
\then
        \dimen6=\shortenproofleft \advance\dimen6 .5\dimen0
        \setbox1=\vbox to\proofdotseparation{\vss\hbox{$\cdot$}\vss}%
        \setbox0=\hbox{%
                \advance\dimen6-.5\wd1
                \kern\dimen6
                $\vcenter to\proofdotnumber\proofdotseparation
                        {\leaders\box1\vfill}$%
                \unhbox\proofrulename}%
\else   \dimen6=\fontdimen22\the\textfont2 
        \dimen7=\dimen6
        \advance\dimen6by.5\proofrulebreadth
        \advance\dimen7by-.5\proofrulebreadth
        \setbox0=\hbox{%
                \kern\shortenproofleft
                \ifdoubleproof
                \then   \hbox to\dimen0{%
                        $\mathsurround0pt\mathord=\mkern-6mu%
                        \cleaders\hbox{$\mkern-2mu=\mkern-2mu$}\hfill
                        \mkern-6mu\mathord=$}%
                \else   \vrule height\dimen6 depth-\dimen7 width\dimen0
                \fi
                \unhbox\proofrulename}%
        \ht0=\dimen6 \dp0=-\dimen7
\fi
\let\doll\relax
\ifwasinsideprooftree
\then   \let\VBOX\vbox
\else   \ifmmode\else$\let\doll=$\fi
        \let\VBOX\vcenter
\fi
\VBOX   {\baselineskip\proofrulebaseline \lineskip.2ex
        \expandafter\lineskiplimit\ifproofdots0ex\else-0.6ex\fi
        \hbox   spread\dimen5   {\hfi\unhbox\proofabove\hfi}%
        \hbox{\box0}%
        \hbox   {\kern\dimen2 \box\proofbelow}}\doll%
\global\dimen2=\dimen2
\global\dimen3=\dimen3
\egroup 
\ifonleftofproofrule
\then   \shortenproofleft=\dimen2
\fi
\shortenproofright=\dimen3
\onleftofproofrulefalse
\ifinsideprooftree
\then   \hskip.5em plus 1fil \penalty2
\fi
}

\usepackage{proof}

\newtheorem{theorem}{Theorem}[section]
\newtheorem{example}{Example}[section]
\newtheorem{remark}{Remark}[section]
\newtheorem{corollary}{Corollary}[section]
\theoremstyle{definition}
\newtheorem{definition}{Definition}

\newtheorem{proposition}{Proposition}

\newtheorem{comment}{Comment}

\newcommand{\B}[1]{\mathbf{#1}}
\newcommand{\BB}[1]{\mathbb{#1}}
\newcommand{\C}[1]{\mathcal{#1}}
\newcommand{\F}[1]{\mathfrak{#1}}

\newcommand{\RM}[1]{\mathrm{#1}}
\newcommand{\SF}[1]{\mathsf{#1}}

\newcommand{\Met}{\mathsf{Met}}
\newcommand{\Mod}{\Lawv\mathsf{Mod}}
\newcommand{\GMet}{\Lawv\mathsf{CCat}}

\newcommand{\colim}{\mathrm{colim}}

\newcommand{\Sym}{\mathrm{Sym}}
\newcommand{\matr}[1]{\hat{#1}}

\newcommand\pfun{\mathrel{\ooalign{\hfil$\mapstochar\mkern5mu$\hfil\cr$\to$\cr}}}

\newcommand{\lam}{\lambda}

\newcommand{\STLC}{\RM{STLC}}
\newcommand{\BSTLC}{\mathsf b\RM{STLC}}

\newcommand{\STDLC}{\RM{ST}\partial\RM{LC}}

\newcommand{\Der}{\SF D}

\newcommand{\Diff}[2]{\Der[#1,#2]}

\newcommand{\true}{\prog{True}}
\newcommand{\false}{\prog{False}}

\newcommand{\Te}[1]{\C T(#1)}

\newcommand{\prog}[1]{\mathtt{#1}}

\newcommand{\Lawv}{\BB L}
\newcommand{\QualREL}[1]{#1 \SF{Rel}}
\newcommand{\QREL}{\QualREL{Q}}
\newcommand{\LREL}{\QualREL{\Lawv}}

\newcommand{\op}{\mathrm{op}}

\newcommand{\menus}{\dotdiv} 

\newcommand{\norm}[1]{\lVert#1\rVert}
\newcommand{\supnorm}[1]{\lVert#1\rVert_\infty}
\newcommand{\absv}[1]{\left\lvert#1\right\rvert}

\newcommand{\trop}[1]{\SF t #1}
\newcommand{\model}[1]{\llbracket#1\rrbracket}

\newcommand{\HOM}[3]{{#1}(#2,#3)}
\newcommand{\N}{\BB N}
\newcommand{\R}{\BB R}
\newcommand{\set}[1]{\{#1\}}
\newcommand{\multiset}{\C M_{\mathrm{fin}}}

\addto\extrasenglish{
  
}

\title{Tropical Mathematics and the Lambda-calculus I:\\
{\Large Metric and Differential Analysis of Effectful Programs}
} 

\author{
Davide {Barbarossa}\\
{\small Dipartimento di Informatica - Scienza e Ingegneria}, \\
{\small Universit\`a di Bologna, Italy}\\
{\small \texttt{davide.barbarossa@unibo.it}}\\
\ \\
Paolo {Pistone}\\
{\small Dipartimento di Informatica - Scienza e Ingegneria}, \\
{\small Universit\`a di Bologna, 
Italy}\\
{\small \texttt{paolo.pistone2@unibo.it}}
}
\date{}

\begin{document}

\maketitle

\begin{abstract}
We study the interpretation of the lambda-calculus in a framework based on tropical mathematics, and we show that it provides a unifying framework for two well-developed quantitative approaches to program semantics: on the one hand program metrics, based on the analysis of program sensitivity via Lipschitz conditions, on the other hand resource analysis, based on linear logic and higher-order program differentiation. 
To do that we focus on the semantics arising from the relational model weighted over the tropical semiring, and we discuss its application to the study of “best case” program behavior for languages with probabilistic and non-deterministic effects. Finally, we show that a general foundation for this approach is provided by an abstract correspondence between tropical algebra and Lawvere’s theory of generalized metric spaces.
\end{abstract}

\section{Introduction}

In recent years, more and more interest in the programming language community has been directed towards the study of \emph{quantitative} properties of programs like e.g.~the number of computation steps or the probability of convergence, 
as opposed to purely \emph{qualitative} properties like termination or program equivalence. 
Notably, a significant effort has been made to extend, or adapt, well-established qualitative methods, like type systems, relational logics or denotational semantics, to account for quantitative properties. We can mention, for example, 
intersection type systems aimed at capturing time or space resources \cite{decarvalho2018, Accattoli2022} or convergence probabilities \cite{Breuvart2018, PistoneLICS2022},  relational logics to account for probabilistic properties like e.g.~differential privacy \cite{Barthe_2012} or metric preservation \cite{Reed2010, dallago}, as well as the study of denotational models for 
probabilistic \cite{Ehrhard2011, Staton2017} or differential \cite{difflambda} extensions of the $\lambda$-calculus. 
The main reason to look for methods relying on (quantitative extensions of) type-theory or denotational semantics is that these approaches yield \emph{modular} and \emph{compositional} techniques, that is, allow one to deduce properties of complex programs from the properties of their constituent parts.   

\subparagraph*{Two approaches to quantitative semantics}

Among such quantitative approaches, two have received considerable attention.

On the one hand one one could mention the approach of \emph{program metrics} \cite{Reed2010, Gaboardi2017, Gabo2019} and \emph{quantitative equational theories} \cite{Plotk}: when considering probabilistic or approximate computation, rather than asking whether two programs compute \emph{the same} function, it makes more sense to ask   whether they compute functions which do not differ \emph{too much}. This has motivated the study of denotational frameworks in which types are endowed with a metric, measuring similarity of behavior; this approach has found  applications in e.g.~differential privacy \cite{Reed2010} and coinductive methods \cite{Bonchi2018}, and was recently extended to account for the full $\lambda$-calculus \cite{Geoffroy2020, PistoneLICS, PistoneFSCD2022}.

On the other hand, there is the approach based on \emph{differential} \cite{difflambda} or \emph{resource-aware} \cite{Boudol1993} extensions of the $\lambda$-calculus, which is well-connected to the so-called \emph{relational semantics} \cite{Manzo2012, Manzo2013, dill} and has a syntactic counterpart in the study of \emph{non-idempotent} intersection types \cite{decarvalho2018, Mazza2016}. This family of approaches have been exploited to account for higher-order program differentiation \cite{difflambda}, to establish reasonable \emph{cost-models} for the $\lambda$-calculus \cite{Accattoli2021}, and have also been shown suitable to the probabilistic setting \cite{Manzo2013, Breuvart2018, PistoneLICS2022}.

In both approaches the notion of \emph{linearity}, in the sense of linear logic \cite{girardLl} (i.e.~of using inputs exactly once), plays a crucial role.
In metric semantics, linear programs correspond to \emph{non-expansive} maps, i.e.~to functions that do not increase distances, and the possibility of duplicating inputs leads to interpret programs with a fixed duplication bound as \emph{Lipschitz-continuous} maps \cite{Gaboardi2017}.
By contrast, in the standard semantics of the differential $\lambda$-calculus, linear programs correspond to linear maps, in the usual algebraic sense, while the possibility of duplicating inputs leads to consider functions defined as \emph{power series}.

A natural question, at this point, is whether these two apparently unrelated ways of interpreting linearity and duplication can be somehow reconciled. At a first glance, there seems to be a  ``logarithmic'' gap between the two approaches:
in metric models a program duplicating an input $n$ times yields a {linear} (hence \emph{Lipschitz}) function $n x$, whereas in differential models it would lead to a {polynomial} function $x^{n}$, thus not Lipschitz. The fundamental idea behind this work is the observation that 
this gap is naturally overcome once we interpret these functions in the framework of tropical mathematics, where, as we will see, the monomial $x^{n}$ precisely reads as the linear function $n x$.

\subparagraph*{Tropical mathematics and program semantics } 

Tropical mathematics was introduced in the seventies by the Brazilian mathematician Imre Simon \cite{Simon} as an alternative approach to algebra and geometry where the usual ring structure of numbers based on addition and multiplication is replaced by the semiring structure given, respectively, by ``$\min$'' and ``$+$''.
%
%
For instance, the polynomial $p(x,y)=x^{2}+xy^{2}+y^{3}$, when interpreted over the tropical semiring, translates as the piecewise linear function
$
\varphi(x,y)=\min\{2x, x+2y, 3y\}
$.
In the last decades, tropical geometry evolved into a vast and rich research domain, providing a combinatorial counterpart of usual algebraic geometry, with important connections with optimisation theory \cite{Sturmfelds}.
Computationally speaking, working with $\min$ and $+$ is generally easier than working with standard addition and multiplication; for instance, the fundamental (and generally intractable) problem of finding the roots of a polynomial admits a \emph{linear time} algorithm in the tropical case (and, moreover,  the tropical roots can be used to approximate the actual roots \cite{Noferini2015}).
The combinatorial nature of several methods in tropical mathematics explains why these are so widely applied in computer science, notably for convex analysis and machine learning (see \cite{Maragos2021} for a recent survey).

Coming back to our discussion on program semantics, tropical mathematics seems to be just  what we look for, as it turns polynomial functions like $x^{n}$ into Lipschitz maps like $n x$.
At this point, it is worth mentioning that a tropical variant of the usual relational semantics of linear logic and the $\lambda$-calculus has already been considered \cite{Manzo2013}, and shown capable of capturing \emph{best-case} quantitative properties, but has not yet been studied in detail. Furthermore, connections between tropical linear algebra and metric spaces have also been observed \cite{Fuji} within the abstract setting of \emph{quantale-enriched} categories \cite{Hofmann2014, Stubbe2014}.
However, a thorough investigation of the interpretation of the $\lambda$-calculus within tropical mathematics and of the potentialities of its applications has not yet been undertaken. 

In this paper we make a first step in such direction, by demonstrating that the relational interpretation of the $\lambda$-calculus based on tropical mathematics does indeed provide the desired bridge between differential and metric semantics, and suggests new combinatorial methods to study probabilistic and non-deterministic programs.

\subparagraph*{Contributions and outline of the paper}


Our contributions in this paper are the following:
\begin{itemize}

\item We first show that tropical polynomials naturally arise in the {best-case} analysis of probabilistic and non-deterministic programs, turning the study of quantitative program behavior into a purely combinatorial problem. This is in Sections \ref{section2} and \ref{section22}.

\item We study the relational model over the tropical semiring, which provides a semantics of effectful extensions of the simply typed $\lambda$-calculus ($\STLC$ in the following) and $\mathrm{PCF}$ \cite{Plotkin1977}. Notably, 
 we show that higher-order programs are interpreted by a  
generalizations of {tropical power series} \cite{Porzio2021}, and we show that these functions are locally Lipschitz-continuous, thus yielding a full-scale metric semantics. This is in Sections \ref{section3} and~\ref{sec:tls}.
 
 \item We exploit the differential structure of the relational model to study the \emph{tropical Taylor expansion} of a $\lambda$-term, which can be seen as an approximation of the term by way of Lipschitz-continuous maps, and we show that it can be used to compute approximated Lipschitz-constants for higher-order programs.
This is in Section \ref{sec:TayLip}.


\item We conclude 
by framing the connection between the 
tropical, differential and metric viewpoints at a more abstract level.
We recall a well-known correspondence between Lawvere's \emph{generalized metric spaces} \cite{Lawvere1973, Stubbe2014} and modules over the tropical semi-ring \cite{Russo2007} and we show that it yields a model of the differential $\lambda$-calculus which extends the tropical relational model. This is in Section~\ref{sec:GMS}.
\end{itemize}
%
%
%
%
%
%
%
%
%

\section{A Bridge between Metric and Differential Aspects}\label{section5bis}

In this section, we discuss in some more detail the two approaches to quantitative semantics we mentioned in the Introduction, at the same time providing an overview of how we aim at bridging them using tropical mathematics.

\subparagraph*{Metric Approach: Bounded $\lambda$-Calculus
}

In many situations (e.g.~when dealing with computationally difficult problems) one does not look for algorithms to compute a function \emph{exactly}, but rather to approximate it (in an efficient way) within some error bound. In other common situations (e.g.~in differential privacy \cite{Alvim2011, Reed2010}) one needs to verify that an algorithm is not \emph{too sensitive} to errors, that is, that a small error in the input will produce a comparably small error in the output. 
In all these cases, it is common to consider forms of denotational semantics in which types are endowed with a \emph{behavioral metric}, that is, a metric on programs which accounts for differences in behavior. 
A fundamental insight coming from this line of work is that, if one can somehow   \emph{bound} the number of times that a program may duplicate its input, the resulting program will be \emph{Lipschitz-continuous}:  
if $M$ may duplicate at most $L$ times, then an error $\epsilon$ between two inputs will result in an error less or equal to $L\cdot \epsilon$ in the corresponding outputs \cite{Reed2010, Gaboardi2017}.
For instance, the higher-order program $M=\lambda f.\lambda x.f(f(x))$, which duplicates the functional input $f$, yields a $2$-Lipschitz map between the metric space $\BB R\multimap \BB R$ of non-expansive real functions and itself: if $f,g$ are two non-expansive maps differing by at most $\epsilon$ (i.e.~for which $|f(x)-g(x)|\leq \epsilon$ holds for all $x\in \BB R$), then the application of $M$ to $f$ and $g$ will produce two maps differing by at most $2\epsilon$. 

%

These observations have led to the study of $\lambda$-calculi with \emph{graded} exponentials, like $\mathsf{Fuzz}$ \cite{Reed2010}, inspired from Girard's Bounded Linear Logic \cite{Girard92tcs}, which have been applied to the study of differential privacy \cite{Gaboardi2013, Gaboardi2017}. The types of such systems are defined by combining linear constructors with a \emph{graded linear exponential comonad} $!_{r}(-)$ \cite{Katsumata2018}.

Yet, what about the good old, ``unbounded'', simply typed $\lambda$-calculus? Actually, by using unbounded duplications, one might lose the Lipschitz property. For instance, while the functions $M_{k}=\lambda x. k\cdot x: \BB R\to \BB R$ are all Lipschitz-continuous, with Lipschitz constant $k$, the function $M=\lambda x.x^{2}$ obtained by ``duplicating'' $x$ is not Lipschitz anymore: $M$ is, so to say, \emph{too} sensitive to errors. 
More abstractly, it is well-known that the category $\Met$ is \emph{not} cartesian closed, so it is not a model of $\STLC$ (yet, several cartesian closed \emph{sub-}categories of $\Met$ do exist, see e.g.~\cite{Clementino2006, PistoneFSCD2022}).
Still, one might observe that the program $M$ above is actually Lipschitz-continuous, if not globally, at least \emph{locally} (i.e.~over any compact set). Indeed, some cartesian closed categories of locally Lipschitz maps have been produced in the literature \cite{Ehrhard2011, PistoneLICS}, and a new example will be exhibited in this paper.

\subparagraph*{Resource Approach: the Differential $\lambda$-Calculus
}

A different family of approaches to linearity and duplication arises from the study of the \emph{differential $\lambda$-calculus} \cite{difflambda} (and differential linear logic \cite{dill}) and its categorical models. 
The key ingredient is a \emph{differential constructor} $\Der[\_,\_]$,  added to the usual syntax of the $\lambda$-calculus. The intuition is that, given $M$ of type $A\to B$ and $N$ of type $A$, the program $\Der[M,N]$, still of type $A\to B$, corresponds to the \emph{linear application} of $M$ to $N$: this means that $N$ is passed to $M$ so that the latter may use it exactly once.
This is also why $\Der[M,N]$ still has type $A\to B$, since $M$ might need \emph{other} copies of an input of type $A$. 
In particular, the application of $\Der[M,N]$ to an ``error term'' $0$ ensures that $M$ will use $N$ exactly once (we say \emph{linearly}).
%

The reason why $\Der$ is called a ``differential'', is twofold: semantically, its interpretation is a generalisation of the usual differential form analysis (see \autoref{section3}); syntactically, it allows to define the so-called \emph{Taylor expansion} $\C T$ of programs:
the idea is that one can expand any application $MN$ as an infinite formal sum of \emph{linear} applications
$\Der^{k}[M,N^k]0$, i.e.~where $N$ is linearly passed exactly $k$ times to $M$; 
doing this recursively gives rise to the suggestive {Taylor formula} $\Te{MN} :=  \sum_{k=0}^{\infty}\frac{1}{!k}\cdot \Der^{k}[\Te{M},\Te{N}^k]0$.
In other words, unbounded duplications correspond to some sort of limit of bounded, but arbitrarily large, ones.

\subparagraph*{Tropical Mathematics: Lipschitz Meets Taylor}

At this point, as the Taylor formula decomposes an unbounded application as a limit of bounded ones, one might well ask whether it could be possible to see this formula as interpreting  a $\lambda$-term 
as a limit of Lipschitz maps, in some sense, thus bridging the metric and differential approaches.  
Here, a natural direction to look for is the \emph{weighted relational semantics} \cite{Manzo2013}, 
due to its strict relations with the Taylor expansion of programs.
However, in this semantics, arbitrary terms correspond to power series, and terms with bounded applications correspond to {polynomials}, hence in any case to functions which are \emph{not} Lipschitz. 

Yet, what if such polynomials were tropical ones, i.e.~piecewise linear functions? This way, the Taylor formula could really be interpreted as a decomposition of $\lambda$-terms via limits (indeed, $\inf$s) of Lipschitz maps. In other words, unbounded term application could be seen 
as a limit of \emph{more and more sensitive} operations.

This viewpoint, that we develop in the following sections, leads to the somehow unexpected discovery of a bridge between the metric and differential study of higher-order programs.
This connection not only suggests the application of optimization methods based on tropical mathematics to the study of the $\lambda$-calculus and its quantitative extensions, but it scales to a 
more abstract level, leading to introduce a 
differential operator for continuous functors between \emph{generalized} metric spaces (in the sense of \cite{Lawvere1973}), as shown in Section \ref{sec:GMS}.

\section{Tropical Polynomials and Power Series}\label{section2}


At the basis of our approach is the observation that the \emph{tropical semiring} $([0,\infty], \min, +)$, which is at the heart of tropical mathematics, coincides with the \emph{Lawvere quantale} $\Lawv=([0,\infty], \geq, +)$ \cite{Hofmann2014, Stubbe2014}, the structure at the heart of the categorical study of metric spaces initiated by Lawvere himself \cite{Lawvere1973}.
Let us recall that a quantale is a complete lattice endowed with a continuous monoid action. In the case of $\Lawv$ the lattice is defined by the reverse order $\geq$ on $\BB R$, and the monoid action is provided by addition. Notice that the lattice join operation of $\Lawv$ coincides with the idempotent semiring operation $\min$. 



Power series and polynomials over the tropical semiring are defined as follows:

\begin{definition}
A \emph{tropical power series} (tps) 
in $k$-variables is a function $f:\Lawv^k\to\Lawv$ of shape $f(x)=\inf_{i\in I}\{ix + \matr f(i) \}$, where 
$I\subseteq\mathbb N^{k}$, $i x$ is the scalar product and $\matr f\in \Lawv^{{\mathbb N}^{k}}$ is a vector of coefficients.
When $I$ is finite, $f$ is called a \emph{tropical polynomial}. 
\end{definition}
Hence, a unary tps is a function $f:\Lawv\to \Lawv$ of the form $f(x)=\inf_{i\in I}\{ix+a_{i}\}$, with $I\subseteq\mathbb N$ and the $a_{i}\in \Lawv$.
In Section \ref{section3} we also consider tps in \emph{infinitely many} variables.

A tropical polynomial is always a piece-wise linear function since, e.g.\ in one variable, it has shape $f(x)=\min_{0\leq j\leq n}\{i_{j}x+c_{i_{j}}\}$.
For example, the polynomials $\varphi_{n}(x)=\min_{0\leq j\leq n}\{jx+2^{-j}\}$
are illustrated in Fig.~\ref{fig:plot1} for $0\leq n \leq 4$.

%
%
%
%
%
%
%
%

\begin{wrapfigure}{r}{0.5\textwidth}
\begin{tikzpicture}[scale=0.9]
\begin{axis}[samples=250]
\addplot[yellow,domain=0:0.8] {1+0.02};

\addplot[orange,domain=0:0.8] {min(x+1/2, 1)+0.01};

\addplot[red,domain=0:0.8] {min(2*x+1/4, x+1/2, 1};
\addplot[blue,domain=0:0.8] {min(3*x+1/8,2*x+1/4, x+1/2, 1)-0.01};
\addplot[orange,domain=0:0.8] {min(4*x+1/16,3*x+1/8,2*x+1/4, x+1/2, 1)-0.02};

\addplot[violet,domain=0:0.8] {min(
10*x+1/1424,
9*x+1/712,
8*x+1/356,
7*x+1/128,
6*x+1/64,
5*x+1/32,
4*x+1/16,3*x+1/8,2*x+1/4, x+1/2, 1)-.03};

\end{axis}

\end{tikzpicture}
\caption{\small Tropical polynomials $\varphi_0,\dots,\varphi_4$ (top to bottom), and the limit tLs $\varphi$ (in violet). The points where the slope changes are  the tropical roots of $\varphi$, i.e.~the points $x=2^{-(i+1)}$, satisfying $ix+2^{-i}=(i+1)x+2^{-(i+1)}$.}
\label{fig:plot1}
\end{wrapfigure}
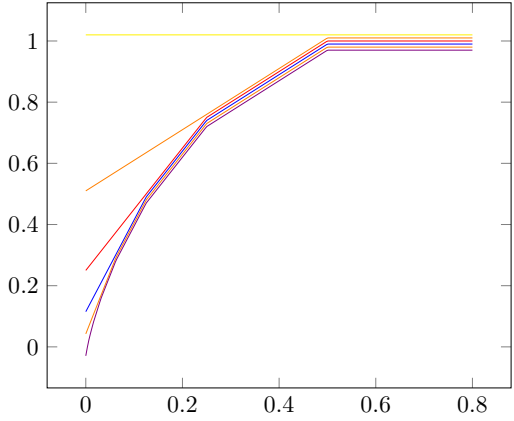 

A \emph{tropical root} of a tps $\varphi$ is a point $x\in \Lawv$ where $\varphi$ is not differentiable (i.e.~where the slope of $\varphi$ changes). When $\varphi$ is a polynomial, the roots of $\varphi$ coincide with the points where the minimum defining $\varphi$ is attained at least twice (see \autoref{fig:plot1}).
Unlike in standard algebra, tropical roots of tropical polynomials can be computed in linear time \cite{Noferini2015}.

%

While tropical polynomials are essentially combinatorial objects, this cannot be said for tps: since $\inf$s are not in general $\min$s, a tps is a ``limit'' of tropical polynomials of higher and higher degree, and its behavior is in general way more difficult to study than that of tropical polynomials~\cite{Porzio2021}. 
E.g., the tLs $\varphi(x):=\inf_{n\in\N}\set{nx+2^{-n}}$ (see Fig.~\ref{fig:plot1}) 
is the ``limit'' of the polynomials $\varphi_{n}$.

%
%



\begin{remark}\label{rmk:val_trop}
 There is a well-known relation between tropical polynomial/power series and usual polynomials/power series. 
Given $Q,L$ semirings with units and zeros, $L$ commutative idempotent, if $L$ is totally ordered by the partial order: $\alpha \preceq \beta$ iff $\alpha +  \beta = \beta$, then one defines a \emph{valuation} \cite{Izhakian2015} to be a map $\mathrm{val}:Q\to L$ s.t.\ $\mathrm{val}(0)=0$, $\mathrm{val}(1)=1$, $\mathrm{val}(ab)=\mathrm{val}(a)\mathrm{val}(b)$, $\mathrm{val}(a+b)\preceq \max\set{\mathrm{val}(a),\mathrm{val}(b)}$.
For example, the valuation which gives $1\in L$ on all the invertibles of $Q$ and $0$ otherwise, is called the \emph{trivial valuation} $\mathrm{val}_1$.

Now, for a power series $f(x)=\sum_{i\in I} a_i x^i$ (polynomials being the case for $I$ finite) with coefficients $a_i$ in a semiring $Q$, and for a fixed valuation $\mathrm{val}:Q\to \Lawv$, one defines the \emph{tropicalisation} $\trop^{\mathrm{val}}f:\Lawv \to \Lawv$ of $f$ as the tropical polynomial/power series function $\trop^{\mathrm{val}}f(\alpha):=\inf_{i\in I} \set{\mathrm{val}(a_i)+i\alpha}$.
\end{remark}

Notice that in $\Lawv$ we have that $\preceq$ is $\geq$ and $\max^{\preceq}=\min$ and the zero is $\infty$.
Therefore, for instance, the tropicalisation (w.r.t.\ any valuation) sends the infinite power series $\sum_{n\geq 2} x^n$ to the tps $\varphi(x)=\inf_n (n+2)x$.
But the latter always coincides with the tropical polynomial $\varphi(x)=2x$, since $nx\geq 0$ for all $n\in\N$, which is in turn the tropicalisation of the polynomial $x^2$.
This shows that tropicalisation is \emph{not} in general an injective operation and in fact, as we show in \autoref{theorem:fepsilon} below, tps (in finitely many variables) have a tendency to collapse, if not globally at least locally, onto tropical polynomials.

For instance, by looking at Fig~\ref{fig:plot1} it appears that, \emph{far from $0$}, $\varphi$ behaves like some of the polynomials $\varphi_{n}$.
In particular, 
$\varphi$ coincides on $[\epsilon,\infty]$ with $\varphi_{n}$,
for $\epsilon \geq 2^{-(n+1)}$ (the smallest tropical root of $\varphi_{n}$).
However, at
%
 $x=0$ we have that $\varphi(x=0)=\inf_{n\in\N} 2^{-n}=0$, and this is the only point where the $\inf$ is \emph{not} a $\min$.
Also, while the derivative of $f$ is bounded on all $(0,\infty)$, for $x\to 0^+$ it tends to $\infty$.
In fact, this is a general phenomenon, as showed below:
%

\begin{theorem}\label{theorem:fepsilon}
For all tps $f(x)=\inf_{n\in \mathbb N^{k}}\{n x+\matr f(n)\}$, for all $0<\epsilon<\infty$, there is a \emph{finite} $\C F_\epsilon \subseteq \N^k$ such that 
%
$f$ coincides on all $[\epsilon,\infty]^k$ with $P_\epsilon( x):=\min_{n\in \C F_\epsilon}\set{n x+\matr f(n)}$.
\end{theorem}
As we'll see, the potential of collapsing infinitary objects (i.e.~tps) into combinatorial ones (i.e.~tropical polynomials), is one of the most intriguing features of tropical semantics. 

For the interested reader, we provide the proof of Theorem \ref{theorem:fepsilon} below. Let us first set the following:
\begin{definition}
 Let $\preceq$ be the product order on $\N^k$ (i.e.\ for all $ m  , n  \in \N^{K}$, $ m  \preceq  n  $ iff $m_{i}\leq n_{i}$ for all $1\leq i\leq K$).
 Of course $ m  \prec  n  $ holds exactly when $ m  \preceq  n  $ and $m_{i}<n_{i}$ for at least one $1\leq i\leq K$.
 Finally, we set $ m  \prec_{1} n  $  iff
$ m  \prec  n  $ and $\sum_{i=1}^{K}n_{i}-m_{i}=1$ (i.e.\ they differ on exactly one coordinate).
\end{definition}

\begin{remark}\label{rmk:AC}
If $U\subseteq \N^{K}$ is infinite, then $U$ contains an infinite ascending chain $ m  _{0}\prec  m  _{1} \prec  m  _{2} \prec \dots$.
This is a consequence of K\"onig Lemma (KL): consider the directed acyclic graph $(U,\prec_{1})$, indeed a $K$-branching tree; if there is no infinite ascending chain $  m  _{0}\prec  m  _{1} \prec  m  _{2} \prec \dots$, then in particular there is no infinite ascending chain $  m  _{0}\prec_{1}  m  _{1} \prec_{1}  m  _{2} \prec_{1} \dots$ so the tree $U$ has no infinite ascending chain; then by KL it is finite, contradicting the assumption. 
\end{remark}

\begin{proof}[Proof of Theorem \ref{theorem:fepsilon}]
We will actually show the existence of $\C F_\epsilon \subseteq_{\mathrm{fin}} \N^k$ such that:
\begin{enumerate}
 \item if $\mathcal{F}_\epsilon= \emptyset$ then $f( x ) = +\infty$ for all $ x \in \Lawv^k$;
 \item if $f( x _0) = +\infty$ for some $ x _0\in [\epsilon,\infty)^{K}$ then $\mathcal{F}_\epsilon= \emptyset$;
 \item the restriction of $f$ on $[\epsilon,\infty]^k$ coincides  with  $P_\epsilon(x):=\min\limits_{n\in \C F_\epsilon}\set{n x+\matr f(n)}$.
\end{enumerate}
Let $\mathcal F_\epsilon$ be the complementary in $\N$ of the set:
\[
 \set{ n  \in\N^{K} \mid \textit{either } \hat f ( n  )=+\infty \textit{ or there is }  m  \prec  n  \textit{ s.t.\ } \hat f( m  )\leq\hat f( n  )+\epsilon}.
\]
In other words, $ n  \in\mathcal F_\epsilon$ iff $\hat f( n  )<+\infty$ and for all $ m  \prec  n  $, one has $\hat f( m  )>\hat f( n  )+\epsilon$.

Suppose that $\mathcal F_\epsilon$ is infinite; then, using Remark~\ref{rmk:AC}, it contains an infinite ascending chain $\set{ m  _0\prec  m  _1\prec\cdots}$.  
By definition of $\mathcal F_\epsilon$ we have then 
$+\infty>\hat f( m  _0)>\hat f( m  _1)+\epsilon>\hat f( m  _2)+2\epsilon>\cdots$, 
so that $+\infty>\hat f( m  _0)>\hat f( m  _{i})+i\epsilon\geq i\epsilon$ for all $i\in\N$.
This contradicts the Archimedean property of $\R$.

1).
We show that if $\mathcal F_\epsilon=\emptyset$, then $\hat f( n  )=+\infty$ for all $ n  \in\N^{K}$.
This immediately entails the desired result.
We go by induction on the well-founded order $\prec$ over $ n  \in\N^{K}$:

- if $ n  =0^{K}\notin\mathcal F_\epsilon$, then $\hat f( n  )=+\infty$, because there is no $ m  \prec n  $.

- if $ n  \notin\mathcal F_\epsilon$, with $ n  \neq 0^{K}$ then either $\hat f( n  )=+\infty$ and we are done, or there is $ m  \prec  n  $ s.t.\ $\hat f( m  )\leq \hat f( n  )+\epsilon$.
By induction $\hat f( m  )=+\infty$ and, since $\epsilon<+\infty$, this entails $\hat f( n  )=+\infty$.

2).
If $f( x _0)=+\infty$ with $ x _0\in [\epsilon,\infty)^{K}$, then necessarily $\hat f( n  )=+\infty$ for all $ n  \in\N^{K}$.
Therefore, no $ n  \in\N^{K}$ belongs to $\mathcal F_\epsilon$.

3).
We have to show that $f( x )=P_\epsilon( x )$ for all $ x \in [\epsilon,+\infty]^{K}$.
By 1), it suffices to show that we can compute $f( x )$ by taking the $\inf$, that is therefore a $\min$, only in $\mathcal F_\epsilon$ (instead of all $\N^{K}$).
If $\mathcal F_\epsilon=\emptyset$ then by 2) we are done (remember that $\min\emptyset := +\infty$).
If $\mathcal F_\epsilon\neq\emptyset$, we show that for all $ n  \in\N^{K}$, if $ n   \notin\mathcal F_\epsilon$, then there is $ m  \in\mathcal F_\epsilon$ s.t.\ $\hat f( m  )+ m   x  \leq \hat f( n  )+ n   x $.
We do it again by induction on $\prec_{1}$:

- if $ n  =0^{K}$, then from $\mathbf  n\notin \mathcal F_{\epsilon}$, by definition of $\mathcal F_\epsilon$, we have $\hat f( n  )=+\infty$ (because there is no $ n  '\prec n  $).
So any element of $\mathcal F_\epsilon\neq\emptyset$ works.

- if $ n  \neq 0^{K}$, then we have two cases:
either $\hat f( n  )=+\infty$, in which case we are done as before by taking any element of $\mathcal F_\epsilon\neq\emptyset$.
Or $\hat f( n  )<+\infty$, in which case (again by definition of $\mathcal F_\epsilon$) there is $ n  '\prec n  $ such that $ \hat f( n  ')\leq \hat f( n  )+\epsilon$  ($\star$).
Therefore we have (remark that the following inequalities hold also for the case $x=+\infty$):
\[\begin{array}{rclr}
 \hat f( n  ')+ n  ' x  & \leq & \hat f( n  ) + \epsilon +  n' x  & \textit{by }(\star) \\
 & \leq & \hat f( n  ) + ( n  - n  ') x  +  n  ' x  & \textit{because $\epsilon\leq\min x $ and $\min  x \leq( n  -  n') x $} \\
 & = & \hat f( n  )+  n   x . &
\end{array}\]
Now, if $ n  '\in\mathcal F_\epsilon$ we are done.
Otherwise $ n  '\notin\mathcal F_\epsilon$ and we can apply the induction hypothesis on it, obtaining an $ m  \in\mathcal F_\epsilon$ s.t.\ $\hat f( m  )+ m   x  \leq \hat f( n  ')+ n  ' x $.
Therefore this $ m  $ works.
\end{proof}


\section{Tropical Semantics and First Order Effectful Programs}\label{section22}

Before discussing how full-scale higher-order programming languages can be interpreted in terms of tropical power series, we highlight how such functions may naturally arise in the study of effectful programming languages.
We will see that, when considering probabilistic and non-deterministic programs, tropical tools can be used to describe the behavior of programs in \emph{the best/worst case}, and may lead to collapse the description of infinitely many possible behaviors into a combinatorial account of the optimal ones.

\subparagraph*{Maximum Likelihood Estimators for Probabilistic Languages}

%

Let us start with a very basic probabilistic language:
the terms are $M::= \true \mid \false \mid M\oplus_p M$, for $p\in[0,1]$, and the operational semantics is $M\oplus_p N\to pM$ and $M\oplus_p N \to (1-p)N$, so that $M\oplus_p N$ plays the role of a probabilistic coin toss of bias $p$.
Consider the program
$
 M:=(\true \oplus_p\false)\oplus_p((\true\oplus_p\false)\oplus_p(\false\oplus_p\true)).
 $
 Calling $q=1-p$, to each occurrence of $\true$ or $\false$ in $M$, univocally determined by an address
$\omega\in \{l,r\}^{*}$, is associated a monomial $P_{\omega}(p,q)$ which determines the probability of the event ``$M\twoheadrightarrow_{\omega} \true/\false$'', that is, that $M$ reduces to $\true/\false$ according to the choices in $\omega$.
Thinking of $p,q$ as parameters, $P_{\omega}(p,q)$ can thus be read as the \emph{likelihood function} of the event ``$M\twoheadrightarrow_{\omega} \true/\false$''.
 For instance, we have
$P_{rll}(p,q):=qp^2$,
$P_{rrr}(p,q):=q^3$, and 
$P_{rrl}(p,q)=P_{rlr}(p,q):=q^2p$.
The polynomial function $Q_{\true}(p,q):=P_{ll}(p,q)+P_{rll}(p,q)+P_{rrr}(p,q)=p^2+p^2q+q^3$ gives instead the probability of the event ``$M\twoheadrightarrow \true$'', and analogously for $Q_{\false}(p,q):=P_{lr}(p,q)+P_{rrl}(p,q)+P_{rlr}(p,q)=pq+2pq^2$.

This way, the probabilistic evaluation of $M$ is presented as a \emph{hidden Markov model} \cite{Baum1966}, a fundamental statistical model, and notably one to which tropical methods are generally applied \cite{Pachter2004}.
Typical questions in this case would be, for a fixed $\omega_0$:
%
\begin{enumerate}
 \item What is the \emph{maximum likelihood estimator} for the event ``$M\twoheadrightarrow_{\omega_0} \false$''?
 I.e., which is the choice of $p,q$ that maximizes the probability $P_{\omega_0}$?
 \item 
Knowing that $M$ produced $\true$ (similarly for $\false$) , which is the \emph{maximum likelihood estimator} for the event ``$M\twoheadrightarrow_{\omega_0}\false$'', knowing that ``$M\twoheadrightarrow \false$''?
I.e., which is the choice of $p,q$ that makes $\omega_0$ the most likely path among those leading to $\false$ (i.e.\ that maximizes the conditional probability $\BB P(``M\twoheadrightarrow_{\omega_0} \false'' \mid ``M\twoheadrightarrow \false'')$)?
\end{enumerate}

Answering 1) and 2) amounts then at solving a maximization problem related to $P_{\omega_{0}}(p,q)$ or $Q_{\true}(p,q)$. In fact, these problems are more easily solved by passing to the associated \emph{tropical} polynomials.
For 1), the maximum values $x,y$ of $P_{rll}(p,q)$ can be computed by finding the \emph{minimum} values of $\mathsf tP_{rll}(-\log p, -\log q)= -2\log p- \log q$. Notice that the latter is precisely the \emph{negative log-probability} of the event ``$M\twoheadrightarrow_{rll} \false$''. For 2), the maximum values of $Q_{\true}(p,q): [0,1]^{2}\to [0,1]$ can be computed as $e^{-\alpha},e^{-\beta}$, where $\alpha,\beta\in[0,\infty]$ are the \emph{minimum} values of the tropical polynomial 
$\mathsf t Q_{\true}(\alpha,\beta) = \min \{ 2\alpha, 2\alpha+\beta, 3\beta\}$.

As we'll see in Section \ref{section3}, this analysis extends to PCF-style programs. For example, the program $M=\mathbf Y(\lambda x.\true \oplus_{p} x)$ yields the power series $Q_{\true}(p,q)=\sum_{n=0}^{\infty}pq^{n}=\frac{p}{1-q}$ that sums all \emph{infinitely many} ways in which $M$ may reduce to $\true$. Notice that the tropicalised series $\mathsf tQ_{\true}(-\log p,-\log q)=\inf_{n\in \mathbb N}\{-\log p -n\log q\}=-\log p$ collapses onto a single monomial describing the \emph{unique} most likely reduction path of $M$, namely the one that passes through a coin toss only once. 
%

\subparagraph*{Best Case Analysis for Non-Deterministic Languages}

This example is inspired from \cite{Manzo2013}. We consider now a basic non-deterministic language with terms $M::= \true \mid  \mathtt{Gen} \mid  M + M$, with an operation semantics comprising a non-deterministic reduction rule 
$M_{1}+M_{2} \stackrel{\alpha}{\to} M_{i}$ and a generation rule
$\mathtt{Gen}\stackrel{\beta}{\to} \true +\mathtt{Gen}$, 
where in each case the value $\alpha,\beta\in \Lawv$ indicates a \emph{cost} associated with the reduction (e.g.~the estimated clock value for the simulation of each reduction on a given machine model). 
Then, any reduction $\omega: M \twoheadrightarrow N$ of a term to (one of its) normal form is  associated with a tropical monomial $P_{\omega}( \alpha,\beta)$ consisting of the sum of the costs of all reductions in $\omega$. For a given normal form $N$, the reductions $\omega_{i}: M \twoheadrightarrow N$ give rise to a tps $\inf_{i\in I}P_{\omega_{i}}( \alpha,\beta)$. 
For example, consider the non-deterministic term
$
M :=\mathtt{Gen}+  ((\true + \true) + \mathtt{Gen})
$. 
The (infinitely many) reduction paths leading to $\true$ can be grouped as follows:
\begin{itemize}
\item left, then reduce $\mathtt{Gen}$ $n+1$-times, then left;
\item right, then left and then either left or right;
\item right twice, then reduce $\mathtt{Gen}$ $n+1$-times and then left.
\end{itemize}
This leads to the tps 
$\varphi_{M\twoheadrightarrow \true}(\alpha)=\inf_{n\in \mathbb N}\big\{2\alpha+(n+1)\beta, 3\alpha, 3\alpha+(n+1)\beta\big \}= \min\{ 2\alpha+\beta, 3\alpha\}$, which describes all possible behaviors of $M$. Notice that, since $\alpha$ and $\beta$ are always positive, the power series $\varphi_{M\twoheadrightarrow\true}(\alpha)$ is indeed equivalent to the tropical polynomial $\min\{ 2\alpha+\beta, 3\alpha\}$. In other words, of the infinitely many behaviors of $M$, only finitely many have chances to be \emph{optimal}: either left + $\mathtt{Gen}$ + left, or right + left + (left,right). Also in this case, reducing to best-case analysis leads to collapse the infinitary description of \emph{all} behaviors to a purely combinatorial description of the finitely many optimal ones. 

Once reduced $\varphi_{M\twoheadrightarrow\true}$ to a polynomial, the best behavior among these will depend on the values of $\alpha$ and $\beta$, and by studying the tropical polynomial $\varphi_{M\twoheadrightarrow\true}$ one can thus answer questions analogous to 2) above, that is, what are the best choices of costs $\alpha,\beta$ making a \emph{chosen} reduction of $M$ to $\true$ the cheapest one?

\section{Tropical Semantics of Higher-Order Programs}\label{section3}
In this section we first recall a general and well-known construction that yields, for any \emph{continuous} semiring $Q$, a model $\QREL_{!}$ of effectful extensions of the simply typed $\lambda$-calculus and $\mathrm{PCF}$, and we show how, when $Q=\Lawv$, it captures optimal program behavior; moreover, we discuss how this model adapts to \emph{graded} and \emph{differential} variants of $\STLC$. 

\subparagraph*{Linear/Non-Linear Algebra on $Q$-Modules}
For any \emph{continuous} semiring $Q$ (i.e.~a cpo equipped with an order-compatible semiring structure), one can define a category $\QREL$  (\cite{Manzo2013} calls it $Q^\Pi$) of ``$Q$-valued matrices'' as follows: $\QREL$ has sets as objects and set-indexed matrices with coefficients in $Q$ as morphisms, i.e.~$\QREL(X,Y):=Q^{X\times Y}$.
The composition $st\in Q^{X\times Z}$ of $t\in Q^{X\times Y}$ and $s\in Q^{Y\times Z}$ is given by $(st)_{a,c}:=\sum_{b\in Y} s_{b,c}t_{a,b}$ (observe that this series always converges because $Q$ is continuous).
For any set $X$, $Q^X$ is a $Q$-semimodule and we can identify $\HOM{\QREL}{X}{Y}$ with the set of linear maps from $Q^X$ to $Q^Y$, which have shape $f(x)_b:=\sum_{a\in X} \matr f_{a,b}x_a$, for some matrix $\matr f\in Q^{X\times Y}$.
%
Notice that usual linear algebra conventions correspond to work in $\QREL^{\op}$, 
e.g.\ the usual matrix-vector product defines a map $Q^Y\to Q^X$.
Following \cite{Manzo2013, Hofmann2014, Ehrhard2005}, we are instead working with transpose matrices.

$\QREL$ admits a comonad $!$ which acts on objects by taking the finite multisets.
Remember that the coKleisli category $\C C_!$ of a category $\C C$ w.r.t.\ a comonad $!$ is the category whose objects are the same of $\C C$, and $\HOM{\C C_!}{X}{Y}:=\HOM{\C C}{!X}{Y}$, with composition $\circ_!$ defined via the co-multiplication of $!$.
Now, although a matrix $t\in\HOM{\QREL_!}{X}{Y}$ yields a linear map $\Lawv^{!X}\to\Lawv^Y$, by exploiting the coKleisli structure we can also ``express it in the base $X$'', which leads to the \emph{non-linear} map $t^{!}:Q^{X}\to Q^{Y}$ defined by the power series
$
t^{!}(x)=t\circ_{!}x : b \mapsto \sum_{\mu\in !X}t_{\mu,b}\cdot x^{\mu}
$, 
where $x^{\mu}= \prod_{a\in x}x_{a}^{\mu(a)}$. 


When we instantiate $Q=\Lawv$, we obtain the category $\LREL$ of $\Lawv$-valued matrices. As one might expect, this category is tightly related to Lawvere's theory of (generalized) metric spaces. For the moment, let us just observe that a (possibly $\infty$) metric on a set $X$ is nothing but a ``$\Lawv$-valued square matrix'' $d:X\times X\to \Lawv$ satisfying axioms like e.g.~the triangular law.
We will come back to this viewpoint in \autoref{sec:GMS}.

Composition in $\LREL$ 
reads as \ $(st)_{a,c}:=\inf_{b\in Y}\set{s_{b,c}+t_{a,b}}$, and the non-linear maps $t^{!}: \Lawv^{X}\to \Lawv^{Y}$ have shape 
$t^!(x)_b=\inf_{\mu\in !X} \set{\mu x+ t_{\mu,b}}$, where $\mu x:=\sum_{a\in X} \mu(a)x_a$. These correspond to the generalisation of tps with \emph{possibly infinitely} many variables (in fact, as many as the elements of $X$).
By identifying $!\set{*}\simeq \N$ and $\Lawv^{\set{*}}\simeq\Lawv$, the tps generated by the morphisms in $\HOM{\LREL_!}{\set{*}}{\set{*}}$ are exactly the 
usual tps's of one variable. %
For example, the $\varphi$ of \autoref{fig:plot1} is indeed of shape $\varphi=t^!$, for $t\in\Lawv^{!\set{*}\times\set{*}}$, $t_{\mu,*}:=2^{-\# \mu}$.

\begin{remark}
The operation $f\mapsto f^{!}$ turning a matrix into a function is reminiscent of the well-known operation of taking the 
 \emph{convex conjugate} $f^*$ of a function $f$ defined over a vector space (itself a generalization of the Legendre transformation).
 Indeed, let $X$, $Y$ be sets and let $\langle \_,\_\rangle:X\times Y \to \mathbb{R}$.
 For $f:X\to \mathbb R$, let $f^*:Y\to \mathbb R$ be defined by $f^*(y):= \sup_{x\in X}\{\langle x,y\rangle - f(x)\}$.
 Then for $X=!A$, $Y=\Lawv^A$, where $A$ is a set, and $\langle \mu, y \rangle:= \mu y$, we have $f^!(y)=(-f)^*(-y)$ for all $f\in\Lawv^{!A}$.
\end{remark}

\subparagraph*{The $\Lawv$-Weighted Relational Model}

The categories $\QREL$ are well-known to yield a model of both probabilistic and non-deterministic versions of $\mathrm{PCF}$ (see e.g.~\cite{Manzo2013, Pagani2018}), which are called \emph{weighted relational models}.
The interpretation of the simply typed $\lambda$-calculus $\STLC$ in $\LREL$ relies on the fact that all categories $\QREL_{!}$ are cartesian closed \cite{Manzo2013}, with cartesian product and exponential objects acting on objects $X,Y$ as, respectively, $X+Y$ and $!X\times Y$. 
Hence, 
any typable term $\Gamma \vdash M:A$ gives rise to a morphism 
$\model{\Gamma \vdash M:A}\in \LREL_{!}(\model{\Gamma}, \model{A})$, and thus to a generalized tps $\model{\Gamma \vdash M:A}^{!}:\Lawv^{\model{\Gamma}}\to \Lawv^{\model{A}}$. 
\begin{example}\label{ex:zxx}
The evaluation morphism $\mathsf{ev}\in\LREL_{!}((!X\times Y) + X, Y)$ yields the tps  
$\mathsf{ev}^{!}: \Lawv^{!X\times Y}\times \Lawv^{X}\to \Lawv^{Y}$ given by 
$b\in Y \mapsto \mathsf{ev}^{!}(F,x)_{b}= \inf_{\mu,b}\{ F_{\mu,b}+ \mu x\}$. So, for instance, supposing the ground type $o$ of $\STLC$ is interpreted as the singleton set $\{*\}$, and recalling the identification $!\{*\}\simeq \mathbb N$, the interpretation of the term $x:o, z:(o\to o\to o) \vdash zxx:o$, involving two consecutive evaluations, yields the tps $\varphi:\Lawv\times \Lawv^{\mathcal M_{\mathrm{fin}}(\mathbb N\times \mathbb N)}\to \Lawv$ given by $\varphi(x,z)=\inf_{n,n'}\{z_{[(n,n')]}+(n+n')x\}$. 
\end{example}
This interpretation extends to $\mathrm{PCF}$ by interpreting the fixpoint combinator $\mathbf Y$ via the matrices $\mathrm{fix}^{X}= \inf_{n}\big\{\mathrm{fix}^{X}_{n}\big\} \in \Lawv^{  !(!X\times X) \times X  }$, where 
 $\mathrm{fix}^{X}_{0}=0$ and $\mathrm{fix}^{X}_{n+1}= \RM{ev}\circ_{!}\langle \mathrm{fix}_{n}^{X}, \RM{id} \rangle $.

One can easily check, by induction on a typing derivation, that for any program of $\STLC$ or $\mathrm{PCF}$, the associated matrix is \emph{discrete}, that is, its values are included in $\{0,\infty\}$. Indeed, as suggested in Section \ref{section2}, the actual interest of tropical semantics lies in the interpretation of effectful programs. 
%
As the homsets $\LREL_{!}(X,Y)$ are $\Lawv$-modules, it is possible to interpret in it  extensions of $\STLC$ and $\mathrm{PCF}$ comprising $\Lawv$-module operations $\alpha \cdot M$ and $M+N$ \cite{Manzo2013}, by 
letting $\model{\Gamma \vdash \alpha\cdot M:A}=\model{\Gamma \vdash M:A}+\alpha$ and 
$\model{\Gamma \vdash M+N:A}=\min\{\model{\Gamma \vdash M:A},\model{\Gamma \vdash N:A}\}$. 
More precisely, \cite{Manzo2013} considers a language $\mathrm{PCF}^{Q}$ corresponding to $\mathrm{PCF}$ extended with $Q$-module operations, with an operational semantics given by rules $M\stackrel{1}{\to}M'$ for each rule $M\to M'$ of $\mathrm{PCF}$ (here $1$ is the monoidal unit of $Q$) as well as 
$M_{1}+M_{2} \stackrel{1}{\to}M_{i}$ and $\alpha\cdot M \stackrel{\alpha}{\to} M $.
Hence, any reduction $\omega=\rho_{1}\dots \rho_{k}: M \twoheadrightarrow N$ is naturally associated with a weight $\mathsf w(\omega)=\sum_{i=1}^{k}\mathsf w(\rho_{i})\in Q$.
In particular, from [Theorem V.6]\cite{Manzo2013} we deduce the following adequation result:
\begin{proposition}
$\model{\vdash_{\mathrm{PCF}^{\Lawv}} M:\mathrm{Nat}}_{n}=  \inf\big\{ \mathsf w(\omega) \ \big \vert \ \omega : M \to \underline n\big \}$ for all $n\in \mathbb N$.
\end{proposition}

The previous result allows to relate the tropical semantics of a program with its best-case operational behavior.  
Observe that the two examples shown in Section \ref{section2} can easily be rephrased in the language $\mathrm{PCF}^{\Lawv}$. 
For instance, for the probabilistic example, 
one can use the translation $(M\oplus_{p}N)^{\circ}= \min\{M^{\circ}+p, N^{\circ}+(1-p)\}$, so that the reductions $:M\oplus_{p}N\stackrel{p}{\to} M$ and $M\oplus_{p}N\stackrel{1-p}{\to} N$ translate into a sequence of two reductions $(M\oplus_{p}N)^{\circ} \stackrel{0}{\to} M^{\circ}\stackrel{p}{\to} M$ and $(M\oplus_{p}N)^{\circ} \stackrel{0}{\to} N^{\circ}\stackrel{1-p}{\to} N^{\circ}$. 
Let $\mathrm{PPCF}$ (for \emph{probabilistic $\mathrm{PCF}$} \cite{Pagani2018}) be standard $\mathrm{PCF}$ extended with the constructor $M\oplus_{p}N$ ($p\in[0,1]$) and its associated reduction rules. From the above discussion we deduce the following:

\begin{corollary}
Let $\vdash_{\mathrm{PPCF}} M:\mathrm{Nat}$ and $n\in\N$.
Considering its interpretation as a function of $p,1-p$, we have that $\model{\vdash_{\mathrm{PPCF}} M^{\circ}:\mathrm{Nat}}_{n}(-\log(p),-\log(1-p))$ 
is the minimum negative log-probability of any reduction from $M$ to $\underline n$, i.e.\ the negative log-probability of (any of) the (equiprobable) most likely reduction path from $M$ to $n$.
\end{corollary}

Remark that this implies that all solution $p\in[0,1]$ to the equation $-\log \mathsf  w(\omega)= \model{\vdash_{\mathrm{PPCF}^{\Lawv}} M^{\circ}:\mathrm{Nat}}_{n}(-\log(p),-\log(1-p))$ are the values of the probabilistic parameter which make the reduction $\omega$ the most likely.

\begin{remark}
The function $\model{\vdash_{\mathrm{PPCF}} M^{\circ}:\mathrm{Nat}}_{n}(\alpha,\beta)$ is a tps, and Theorem \ref{theorem:fepsilon} ensures that this function coincides \emph{locally} with a tropical polynomial. This means that, for any choice of $p,1-p$, the most likely reduction path of $M$ can be searched for within a \emph{finite} space.
\end{remark}

Finally, \cite{Manzo2013} obtained a similar result for a non-deterministic version of PCF, by translating each term into $\mathrm{PCF}^{\Lawv}$ via $(\lambda x.M)^{\circ}=\lambda x.M^{\circ}+1$ and $(\mathbf YM)^{\circ}= \mathbf Y(M^{\circ}+1)$
(\cite{Manzo2013} considers the discrete tropical semiring $\mathbb N\cup\{\infty\}$, but the result obviously transports to $\Lawv$), and in that case \cite[Corollary VI.10]{Manzo2013} gives that $\model{\vdash M^{\circ}:\mathrm{Nat}}_{n}$ computes the minimum number of $\beta$- and $\mathsf{fix}$- redexes reduced in a reduction sequence from $M $ to $\underline n$. 



\subparagraph*{$\mathbb N$-Graded Types}\label{sec:BSTLC}

We now show how to interpret in $\LREL_{!}$ a graded version of $\STLC$, that we call $\BSTLC$, indeed a simplified version of the well-studied language $\mathrm{Fuzz}$ \cite{Reed2010}.
This language is based on a graded exponential $!_{n}A$, corresponding to the possibility of using an element of type $A$ \emph{at most} $n$ times. In particular, if a function $ \lambda x.M$ of type $!_nA\multimap B$, then, for any $N$ of type $A$, $x$ is duplicated \emph{at most} $n$ times in any reduction of $(\lambda x.M)N $ to the normal form.

Graded simple types are defined by $A::= o \ \mid  \ !_{n}A \multimap A$; the contexts of the typing judgements are sets of declarations of the form $x :_{n}A$, with $n\in \mathbb N$;
given two contexts $\Gamma,\Delta$, we define $\Gamma+\Delta$ recursively as follows: if $\Gamma$ and $\Delta$ have no variable in common, then $\Gamma+\Delta=\Gamma\cup \Delta$; otherwise, we let $(\Gamma, x:_{m} A)+( \Delta, x:_{n} A) =  (\Gamma+\Delta), x:_{m+n}A$. 
Moreover, for any context $\Gamma$ and $m\in \mathbb N$, we let $m\Gamma$ be made all $x:_{mn}A$ for $(x:_{n}A) \in \Gamma$.  
The  typing rules of $\BSTLC$ are illustrated in Fig.~\ref{fig:rules},

\begin{figure}
\fbox{
	\begin{minipage}{0.9\textwidth}
	\begin{center}
	$\prooftree
	\justifies
	x:_{1}A\vdash x: A
	\endprooftree$
	
	\bigskip
	
	\begin{minipage}{0.42\textwidth}
	\begin{center}
	$\prooftree
		\Gamma \vdash M:A
		\justifies
		\Gamma, x:_{0}B \vdash M:A
		\endprooftree 
		$
		
	\bigskip
	
	$\prooftree
		\Gamma, x:_{n} A\vdash M: B
		\justifies
		\Gamma\vdash \lambda x.M: !_{n}A\multimap B
		\endprooftree$
		
		\medskip
				
	\end{center}
	\end{minipage} \ \ \ \ 
	\begin{minipage}{0.42\textwidth}
	\begin{center}
		
		$\prooftree
		\Gamma, x:_{n}B, y:_{m} B\vdash M:A
		\justifies
		\Gamma, x:_{n+m}B\vdash M\{x/y\}:A
		\endprooftree $
		
		\bigskip
		
		$\prooftree
		\Gamma \vdash M: !_nA\multimap B
		\quad
		\Delta\vdash N: A
		\justifies
		\Gamma +n\Delta\vdash MN: B
		\endprooftree
		 $
		 
		 \medskip
		 
	\end{center}
	\end{minipage}
	\end{center}
	\end{minipage}
	}
	\caption{Typing rules for $\BSTLC$.}
	\label{fig:rules}
	\end{figure}

Now, one can see that the comonad $!$ of $\LREL$ can be ``decomposed'' into a family of ``graded exponentials functors'' $!_n:\LREL\to\LREL$ ($n\in\BB N$), where $!_{n}X$ is the set of multisets on $X$ of cardinality \emph{at most} $n$. 
The sequence $(!_n)_{n\in\N}$ gives rise to a so-called \emph{$\N$-graded linear exponential comonad} on (the SMC) $\LREL$ \cite{Katsumata2018}. 
As such, $(\LREL,(!_n)_{n\in\N})$ yields then a model of $\BSTLC$. Remark that arrow types are interpreted via $\model{!_{n}A\multimap B}:= !_{n}\model A \times \model B$. Notice that, whenever $\model *$ is finite, the set $\model{A}$ is finite for any type $A$ of $\BSTLC$.

\subparagraph*{The Differential $\lambda$-Calculus}\label{sec:STDLC}

We recall the interpretation in $\LREL_{!}$ of the simply typed \emph{differential} $\lambda$-calculus $\STDLC$,  
 an extension of $\STLC$ ensuring exact control of duplications. The syntax of $\STDLC$ (see \cite[Section 3]{Manzo2010}) is made of \emph{terms} $M$ and \emph{sums} $\mathbb T$, mutually generated by: $M::= x\mid \lambda x.M \mid M\mathbb T \mid \Diff{M}{M}$ and $\mathbb T::= 0 \mid M \mid M+\mathbb T$,
quotiented by equations that make $+,0$ form a commutative monoid on the set of sums, 
by linearity of $\lam x.(\_)$, $\Diff{\_}{\_}$ and $(\_)\mathbb T$ (but \emph{not} of $M(\_)$) and by irrelevance of the order of consecutive $\Diff{\_}{\_}$.
We follow the tradition of quotienting also for the idempotency of $+$.
The typing rules are illustrated in Figure~\ref{fig:rules2}, where a context $\Gamma$ is a list of typed variable declarations.
The main feature of this language is that $\Der^n[\lambda x.M,N^n]0$ has a non-zero normal form iff $x$ is duplicated exactly $n$ times during reduction.

\begin{figure}
\fbox{
	\begin{minipage}{0.9\textwidth}
	\begin{center}
	\begin{minipage}{0.42\textwidth}
	\begin{center}
	$\prooftree
	\justifies
	\Gamma\vdash 0:A
	\endprooftree$
	
	\bigskip
	
	$\prooftree
		\Gamma, x: A\vdash M: B
		\justifies
		\Gamma\vdash \lambda x.M: A\to B
		\endprooftree $
		
		\bigskip
		
		$ \prooftree
		\Gamma \vdash M: A\to B
		\quad
		\Gamma \vdash N: A
		\justifies
		\Gamma \vdash \Diff{M}{N}: A\to B
		\endprooftree$
		
		\medskip
		
	\end{center}
	\end{minipage} \ \ \ \ 
	\begin{minipage}{0.42\textwidth}
	\begin{center}
	$\prooftree 
	\justifies
	\Gamma, x:A \vdash x: A
	\endprooftree $
	
	\bigskip
	
	$\prooftree
		\Gamma \vdash M: A\to B
		\quad
		\Gamma\vdash \mathbb T: A
		\justifies
		\Gamma \vdash M\mathbb T: B
		\endprooftree $
		
		\bigskip
		
		$\prooftree
		\Gamma\vdash M_1: A
		\,\cdots\,
		\Gamma \vdash M_n:A
		\justifies
		\Gamma \vdash M_1+\cdots +M_n : A
		\using (n\geq 2)
		\endprooftree
		$
		
		\medskip
		
	\end{center}
	\end{minipage}
	\end{center}
	\end{minipage}
	}
	\caption{Typing rules of $\STDLC$.}
	\label{fig:rules2}
	\end{figure}

The categorical models of $\STDLC$ are called \emph{cartesian closed
 differential $\lambda$-categories} (CC$\partial\lambda$C)\cite{Manzo2010,Blute2009, Blute2019}. These are CCCs enriched over commutative monoids (i.e.\ morphisms are summable and there is a $0$ morphism), with the cartesian closed structure compatible with the additive structure, 
and equipped with a certain \emph{differential operator} $D$, turning a morphism $f:A\to B$ into a morphism $Df: A\times A\to B$, and 
generalising the usual notion of differential, see e.g.\ \cite{BluteEhrhTass10}.
An example is the CC$\partial\lambda$C of convenient vector spaces with smooth maps, where $D$ is the ``real'' differential of smooth maps.

Applying \cite[Theorem 6.1]{lemay2020} one can check that
 $\LREL_!$ is a CC$\partial\lambda$C (see Section \ref{section5bis}) when equipped with $D:\HOM{\LREL}{!X}{Y}\to \HOM{\LREL}{!(X\& X)}{Y}$ defined as $(Dt)_{\mu\oplus\rho,b}=t_{\rho+\mu,b}$ if $\#\mu=1$ and as $\infty$ otherwise (using the iso $(\mu,\rho)\in !Z\times !Z'\mapsto\mu\oplus\rho \in !(Z+Z')$).
 The differential operator $D$ of $\LREL_{!}$ translates into a differential operator $D_{!}$ turing a tps $f:\Lawv^{X}\to \Lawv^{Y}$ into a tps $D_{!}f:\Lawv^{X}\times \Lawv^{X}\to \Lawv^{Y}$, linear in its first variable, and given by 
$D_{!}f(x,y)_{b}=\inf_{a\in X, \mu\in !X}\left\{\matr f_{\mu+a}+x_{a}+\mu y\right\}$. One can check that, when $f$ is a tropical polynomial, $D_{!}f$ coincides with the standard tropical derivative (see e.g.~\cite{Grigoriev2017}).

\begin{comment}
Writing $\Der^2[\_,(\_)^2]$ as a shortcut for $\Der[\Der[\_,\_],\_]$ and $\Der^1[\_,(\_)^1]$ for $\Diff{\_}{\_}$, the analogue of the previous $\BSTLC$-term is $\vdash_{\STDLC} \lambda {\color{red}z}. \Der^{\color{red}2}[
	\lambda x{\color{green}y}.
		\Der^{\color{violet}1} [
				\Der^{\color{blue}1} [y, x^{\color{blue}1}]
        0, x^{\color{violet}1}
	]0
, z^{\color{red}2}]0
: {\color{red}*}\to ({\color{green}* \to * \to *}) \to *$.
In particular, if the \emph{multiplicities} of the arguments (the colored exponents) do not exactly match the number of duplications, e.g.\ in $\vdash_{\STDLC} \lambda z. \Der^{\color{red}3}[
	\lambda xy.
		\Diff{
				\Diff{y}{x}
		0}{x}
	0
, z^{\color{red}3}]0
: *\to (* \to * \to *) \to *$, then the term reduces to the empty sum $0$ (representing an \emph{error}).
\end{comment}

\section{On Tropical Power Series}\label{sec:tls}

As seen in Section \ref{section2}, tropical polynomials are piecewise-linear functions, hence concave and Lipschitz-continuous. Moreover, tps in finitely many variables are locally equivalent to tropical polynomials (except at some singular points), and are thus also concave and locally Lipschitz-continuous.
In this section we show that much of these properties extend also to tps with infinitely many variables, as those arising from the tropical relational model, 
and by a tropical polynomial or power series we mean one with possibly infintely many variables.
The literature on tropical power series is often recent (e.g.~\cite{Porzio2021}), and several results we prove in this section are, to our knowledge, new.
Notice that, as a set, $\Lawv^X=[0,\infty]^X$, and with the usual $+$ and $\cdot$ it is a $\BB R_{\geq0}$-semimodule, let us call it $\overline{\BB R}_{\geq 0}^X$.
Together with the usual sup-norm $\norm{x}_\infty:=\sup_{a\in X} x_a$, it can be showed to be a Scott-complete \emph{normed cone} (see~\cite{Selinger2004} or the appendix).
Suitable categories of cones have been recently investigated as models of probabilistic computation (\cite{Crubillie2018, EhrPagTas2018, Ehrhard2020}).
The cone structure of $\overline{\BB R}_{\geq 0}^X$ also induces a partial order on it, its \emph{cone partial-order}:
$x\leq y$ iff $y=x+z$ for some (unique) $z\in\overline{\BB R}_{\geq 0}^X$.
It actually coincides with the pointwise order on $\overline{\BB R}_{\geq 0}$ (and makes it a Scott-continuous dcpo).
In this section we consider tps w.r.t.\ this structure.

\subparagraph*{Continuity of tps}\label{subsec:cont}

Looking at Fig~\ref{fig:plot1}, we see that $\varphi$, just like the polynomials $\varphi_{n}$, is non-decreasing and concave.
This is indeed always the case:

\begin{proposition}\label{prop:nondecr+conc}
 All tps 
are non-decreasing and concave, w.r.t.\ the pointwise order on $\overline{\BB R}_{\geq 0}^{X}$.
\end{proposition}

The tps $\varphi$ is continuous on $\BB R_{\geq0}$ (w.r.t.\ the usual norm of real numbers).
We can generalise this property, dropping the case of $x$ having some $0$ coordinate.
But we have to be careful, because while in the finite dimensional $\BB R^n$, every real convex function is continuous because it is necessarily locally bounded from above (the sup-norm and the euclidean one are equivalent) \cite[Proposition 4.7]{Cobzas2017}, in infinite dimensions the former condition is no longer true \cite[Example 4.8]{Cobzas2017}.
However, \cite[Proposition 4.4.(3)]{Cobzas2017} shows that it is the only requirement to ask: if a real-valued convex function with domain a convex open subset of a locally convex topological $\BB R$-vector space (LCTVS) is, locally around any point, bounded from above by a finite non-zero constant, then it is continuous on all its domain.

\begin{theorem}\label{thm:cont}
 All tps $f:\overline{\BB R}_{\geq 0}^X\to\overline{\BB R}_{\geq 0}$ are continuous on $(0,\infty)^X$, w.r.t.\ to the norm $\norm{\cdot}_\infty$.
\end{theorem}
\begin{proof}
By \autoref{prop:nondecr+conc}, $-f$ is convex.
Since $f\geq 0$ on all $\overline{\BB R}_{\geq 0}^X$, we have e.g.\ $-f\leq 1$ on $\overline{\BB R}_{\geq 0}^X$.
Now $(0,\infty)^X\subseteq \BB R^X$ is open and convex, 
so \cite[Proposition 4.4.(3)]{Cobzas2017} entails the continuity of $-f$ on $(0,\infty)^X$, hence that of $f$ on it.
\end{proof}

In analogy with \cite[Proposition 17]{DanEhrh2011}, we also have:

\begin{theorem}\label{thm:ScottCont}
 All monotone (w.r.t.\ pointwise order) and $\norm{\cdot}_{\infty}$-continuous functions $f:(0,\infty)^X\to (0,\infty)$ are Scott-continuous.
 In particular, all tps $f:\overline{\BB R}_{\geq 0}^X\to\overline{\BB R}_{\geq 0}$ are Scott-continuous on $(0,\infty)^X$ w.r.t.\ the pointwise orders.
\end{theorem}

\subparagraph*{Lipschitz-continuity of tps}\label{sec:4C}


Let us first look at what happens with those tps which are either \emph{linear} or obtained via bounded exponentials.
The result below is in analogy with what happens in the usual metric semantics of Fuzz, where 
linear functions are non-expansive and $n$-bounded functions are $n$-Lipschitz \cite{Reed2010}.

\begin{proposition}\label{prop:troplinear}
\begin{enumerate}
\item If a tps $f: \overline{\BB R}_{\geq 0}^{X}\to \overline{\BB R}_{\geq 0}^{Y}$ arises from a matrix $\matr f:X\times Y\to \overline{\BB R}_{\geq 0}$ (i.e.~it is tropical linear), then $f$ is non-expansive (i.e.\ $1$-Lipschitz).  
\item If  $f: \overline{\BB R}_{\geq 0}^{X}\to \overline{\BB R}_{\geq 0}^{Y}$ arises from a matrix $\matr f: \ !_{n}X\times Y\to \overline{\BB R}_{\geq 0}$, then $f$ is $n$-Lipschitz-continuous.
\end{enumerate}
\end{proposition}
\begin{proof}[Proof sketch]
1). Using the fact that $f( x)_{b}= \inf_{a\in X}\{\matr f_{a,b}+ x_{a}\}$,
the problem reduces to: $|(\matr f_{a,b}- x_{a})- (\matr f_{a,b}- y_{a})| = |x_{a}- y_{a}|\leq \norm{x- y}_{\infty}$. 

2.) Follows from 1.~and the remark that, for all $x\in \Lawv^{X}$, $\norm{ !_{n} x-!_{n} y}_{\infty}\leq n\cdot \norm{ x- y}_{\infty}$, where $!_{n} x$ is the restriction of $! x$ to $\C M_{\leq n}(X)$.%
\end{proof} 

Observe that on the hom-sets $\HOM{\LREL_{!}}{X}{Y}$ there are two natural notions of distance: the metric $\| f-g\|_{\infty}$ arising from the norm and the one arising from the usual $\sup$-metric $d_\infty(f,g):=\sup_{x\in \Lawv^{X}} \|f^{!}(x)-g^{!}(x)\|_{\infty}$.
In general one has $\| f-g\|_{\infty}\geq d_\infty(f,g)$, the equality holding when $f^{!},g^{!}$ are linear (i.e.~when they arise from morphisms of $ \LREL(X,Y)$).
 
%


For any tropical polynomial $\varphi:\overline{\BB R}_{\geq 0}^{X}\to \overline{\BB R}_{\geq 0}^{Y}$, the associated matrix has shape $!_{\mathrm{deg}(\varphi)}(X)\times Y\to \overline{\BB R}_{\geq 0}$ (as a monomial $\mu_ix+c_{i}$ yields a matrix entry on $!_{\#\mu_i}X\times Y$). Hence, using Proposition \ref{prop:troplinear} 2., we have:
\begin{corollary}\label{prop:polylip}
Any tropical polynomial $\varphi:\overline{\BB R}_{\geq 0}^{X}\to\overline{\BB R}_{\geq 0}$ is $\mathrm{deg}(\varphi)$-Lipschitz continuous.
\end{corollary}

We now show that, if we consider the \emph{full} exponential $!$, i.e.~arbitrary tps, we can still prove that a local Lipschitz condition holds. 
In \cite[Theorem 6.4]{Cobzas2017} a locally Lipschitz property is obtained for locally convex topological vector spaces, under the hypothesis of continuity. \cite[Proposition 6.5]{Cobzas2017} shows that continuity is used in order to have a locally bounded condition, the crucial ingredient of the proof.
Instead of showing how our case fits into such theorems, we prefer to state the following theorem, basically a particular case of \cite[Theorem 6.9, Lemma 6.10]{Cobzas2017}:

\begin{theorem}\label{thmTLSlocLip}
All tps $f:\overline{\BB R}_{\geq 0}^X\to\overline{\BB R}_{\geq 0}$ are locally Lipschitz on $(0,\infty)^X$.
Moreover, the Lipschitz constant of $f$ on $\overline{B_{\delta}(x)}$ can be chosen to be $\frac{1}{\delta}\max_{\overline{B_{3\delta}(x)}} f$.
\end{theorem}
\begin{proof}
By observing that $(0,\infty)^X$ is open and convex in $(\BB R^X,\norm\cdot)$, we apply the result that for all $f:V\subseteq (\BB R^X,d) \to (\BB R,\absv \cdot)$ concave and locally bounded, with $V$ open and convex and $d$ any metric, $f$ is locally Lipschitz, with the stated Lipschitz constant on $\overline{B_{\delta}(x)}$. %
\end{proof}

\begin{comment}

\subsection{The tropical interpretation of high-order programs}\label{subsec:tro_int}

Rcall that the Taylor expansion $\Te{MN}$ of $MN$ is the set $\set{\Der^{n}[t,u_1,\dots,u_n]0 \mid n\in\N,t\in\Te{M},u_i\in\Te{N}}$.
$\LREL_!$ is also well-behaved w.r.t.\ the Taylor expansion, as expressed by the following two results.
First, it can be patientely checked that in $(\LREL_!,D)$ all morphisms can be Taylor expanded  (see \cite[Definition 4.22]{Manzo2012}):

\begin{theorem}\label{thm:modelsTaylor}
 For all $t\in\HOM{\LREL_!}{Z}{X\multimap Y}$, $s\in\HOM{\LREL_!}{Z}{X}$ we have:
 \begin{align}\label{eq:taylorcat}
  \RM{ev}\circ_!\langle t,s\rangle =
  \inf\limits_{n\in\N}
  \set{((\dots((\Lambda^- t)\star s)\star \dots)\star s)\circ_! \langle \RM{id},\infty \rangle}.
 \end{align} 
\end{theorem}
Here,
$u\star s= (Du)\circ_{!} \langle \langle  \infty, s\circ_{!} \pi_{1}\rangle,\mathrm{id}\rangle$ corresponds to the application of the derivative of $u$ on $s$, and $\Lambda^-$ is the uncurry operator.
Hence the right-hand term in \eqref{eq:taylorcat} corresponds to the $\inf$ of the $n$-th derivative of $\Lambda^{-}t$ applied to ``$n$ copies'' of $s$,  i.e.~it coincides with the tropical 
version of the usual Taylor expansion.
Second, since $\LREL_!$ has countable sums (all countable $\inf$s converge), an immediate adaptation of the proof of \cite[Theorem 4.23]{Manzo2012} shows:

\begin{corollary}\label{cor:T(M)=M}
 For $\Gamma\vdash_{\STLC} M:A$, we have 
$\model{\Gamma\vdash_{\STLC} M:A}=\inf_{t\in\Te{M}} \model{\Gamma\vdash_{\STDLC} t:A}$. 
\end{corollary}

Let us now see what the results proved in the previous section translate into, when referred to the interpretation of higher-order programs.
Remark that the metric spaces $(\Lawv^{\mathlarger{+}_{i=1}^n X_i},d_{\infty})$ and $(\prod_{i=1}^n \Lawv^{X_i},\max_{i=1}^n d^{X_i}_\infty)$ are trivially isometric, so we identify them.


\begin{corollary}
Let $\model A$ be a finite set.
\begin{enumerate}
\item $\model{\Gamma \vdash_{\BSTLC} M:B}^!:\prod\limits_{(x_i:_{n_i} A_i)\in\Gamma} \!\!\!\!\!\Lawv^{\model{A_i}} \to \Lawv^{\model B}$ is a \emph{tropical polynomial}, thus (as $\model A$ is finite), a \emph{Lipschitz} function.
\item $\model{\Gamma \vdash_{\STLC} M:B}^!, \model{\Gamma \vdash_{\STLC_\oplus} M:A}^!:\prod\limits_{(x_i: A_i)\in\Gamma} \!\!\!\!\!\Lawv^{\model{A_i}} \to \Lawv^{\model B}$ are \emph{locally} Lipschitz maps.
\item $\Te{M}$ decomposes $\model{\Gamma \vdash_{\STLC} M:A}^!$ as an $\inf_{t\in\Te{M}}\model{\Gamma\vdash_{\STDLC} t:A}^!$ of \emph{tropical polynomials}, thus (as $\model A$ is finite), \emph{Lipschitz} functions.
\end{enumerate}
\end{corollary} 
\begin{proof}
1). Since $\model A$ is finite, also $\model *$ is.
Thus, as we already observed, the interpretation of a bounded term is a tropical polynomial.
Now we apply Corollary~\ref{prop:polylip} to each coordinate of the image, and by taking the maximum Lipschitz constant among the finite number $\mathrm{Card}(\model A)$ of them, we obtain the thesis.
2). It follows immediately from Theorem~\ref{thmTLSlocLip} and the fact that $\model A$ is finite.
3). It follows from \autoref{cor:T(M)=M} plus the easily checked fact that, for $(f_n)_{n\in\N}\subseteq\Lawv^{!X\times Y}$, we have $\left(\inf_{n\in\N} f_n\right)^!:\Lawv^X\to\Lawv^Y$, with $\left(\inf_{n\in\N} f_n\right)^!=\inf_{n\in\N} f_n^!$.
\end{proof}

Remark that the restriction $\model A$ finite is without loss of generality, since by Currying all programs can be seen having type $*$, which is natural to interpret as a singleton.

Remark that, if $Y$ is finite and $f:\Lawv^{\mathlarger{+}_{i=1}^n X_i}\to\Lawv^Y$ is $K$-Lipschitz, then $f:\prod_{i=1}^n \Lawv^{X_i}\to \Lawv^Y$ is $K$-Lipschitz on each of the $n$ variable separately.

A consequence of (3) is that the pointwise distance between two interpretations of programs can always be bounded via Lipschitz tropical polynomial approximants of the initial two programs.
\begin{corollary}
 Let $\Gamma \vdash_{\STLC} M:A$ and $\Delta\vdash_{\STLC} N:B$.
 For all $\epsilon>0$, $x\in\Lawv^{\model \Gamma}$, $b\in\model A$, there exist $t\in\Te{M}$, $u\in\Te{N}$ s.t.\ $\big| \model{\Gamma \vdash_{\STLC} M:A}^!(x)_b - \model{\Delta \vdash_{\STLC} N:B}^!(x)_b \big| \leq 2\epsilon + \big| \model{\Gamma \vdash_{\STDLC} t:A}^!(x)_b - \model{\Delta \vdash_{\STDLC} u:B}^!(x)_b \big|$.
\end{corollary}

\end{comment}

\section{Lipschitz Meets Taylor}\label{sec:TayLip}

In this section we finally relate the metric and differential analysis of higher-order programs in the tropical relational model. 

The key ingredient is the notion of Taylor expansion $\Te{M}$ of a $\lambda$-term $M$. This is a set of terms of the differential $\lambda$-calculus defined inductively as: 
%
%
%
%
%
%
%
$\Te{x}=\{x\}$, $\Te{\lambda x.M}=\{\lambda x.t\mid t\in \Te{M}\}$ and 
$\Te{MN}=\{t\cdot \langle u_{1},\dots, u_{k}\rangle  \mid k\in \mathbb N, t\in \Te{M}, u_{i}\in \Te{N} \}$, where $t\cdot \langle u_{1},\dots, u_{k}\rangle$ is an abbreviation for $\mathsf D^{k}[t, u_{1},\dots, u_{k}] 0$. 
Observe that in the terms appearing in $\Te{M}$ all applications are \emph{bounded}: they may use an exact number of copies of their input. 
Such terms are usually called \emph{resource $\lambda$-terms} \cite{Pagani2009, Manzo2012}. One can easily check that for all terms $\Gamma \vdash_{\STLC}M:A$ and  $t\in \Te{M}$, also 
$\Gamma \vdash_{\STDLC}t:A$ holds. 
\begin{example}
Considering the term $M=zxx$ from Example \ref{ex:zxx}, all terms
 $t_{n,m}=z \,\langle x^{n}\rangle\,\langle x^{m}\rangle$, for $n,m\in \mathbb N$, are in $\Te{M}$. 
Notice that the interpretation of $t_{n,m}$ yields a tropical polynomial
$\model{t_{n,m}}^{!}(x)(z)= y_{[n,m]}+(n+m)x$, rather than a tps. 
However, this is not a general fact: consider $y:(o\to o )\to (o\to o), x:(o\to o) \vdash t:  (o\to  o)$
with $t=y\cdot \langle y\cdot \langle x\rangle \rangle\in \Te{y(yx)}$. Then 
$\model{t}^{!}: \Lawv^{  !\mathbb N\times \mathbb N}
\times \Lawv^{\mathbb N} \to \Lawv^{\mathbb N}$ is given by
$
\model{t}^{!}(y,x)_{i}= \inf_{m,n\in \mathbb N}\big\{    y_{[m],i}  +  y_{[n],m}+  x_{n}\}
$, 
which is not a polynomial. Yet, $\model{t}^{!}$ is Lipschitz, more precisely, 1-Lipschitz in $x$ and 2-Lipschitz in $y$. This is a general fact, as shown by Theorem \ref{thm:taylor} below.
\end{example}

%
We have already shown that the tropical differential makes $\LREL_{!}$ a model of the differential $\lambda$-calculus. We now show that it also models the Taylor expansion (this needs not be true for \emph{any} CC$\partial\lambda$C).
First, it can be patiently checked that (see \cite[Definition 4.22]{Manzo2012}):
%
\begin{theorem}\label{thm:modelsTaylor}
Morphisms in $(\LREL_{!},D)$ can be Taylor-expanded: for all $t\in\HOM{\LREL_!}{Z}{!X\multimap Y}$, $s\in\HOM{\LREL_!}{Z}{X}$ we have 
  $\RM{ev}\circ_!\langle t,s\rangle =
  \inf\limits_{n\in\N}
  \set{((\dots((\Lambda^- t)\star s)\star \dots)\star s)\circ_! \langle \RM{id},\infty \rangle}.$
\end{theorem}
The equation above is a tropical reformulation of the Taylor formula from the Introduction:
$u\star s= (Du)\circ_{!} \langle \langle  \infty, s\circ_{!} \pi_{1}\rangle,\mathrm{id}\rangle$ corresponds to the application of the derivative of $u$ on $s$, and $\Lambda^-$ is the uncurry operator.
Hence the right-hand term corresponds to the $\inf$ of the $n$-th derivative of $\Lambda^{-}t$ applied to ``$n$ copies'' of $s$.

Second, since $\LREL_!$ has countable sums (all countable $\inf$s converge), an immediate adaptation of the proof of \cite[Theorem 4.23]{Manzo2012} shows:

\begin{corollary}\label{cor:T(M)=M}
$\model{\Gamma\vdash_{\STLC} M:A}=\inf_{t\in\Te{M}} \model{\Gamma\vdash_{\STDLC} t:A}$. 
\end{corollary}

Using the results of the previous section, as well as the results above, we now deduce the following properties:

\begin{theorem}\label{thm:taylor}
Let $\mathcal S$ be one of $\RM{PCF}^{\Lawv},\STLC,\BSTLC,\STDLC$. Let $\Gamma\vdash_{\mathcal S}M:A$ and $a\in \model{A}$.
\begin{enumerate}
\item For $\mathcal S=\BSTLC$, $\model{\Gamma\vdash_{\mathcal S}M:A}^{!}_{a}$ is a tropical polynomial, and thus Lipschitz;

\item For $\mathcal S=\STDLC$, if $t\in \Te{M}$, then $\model{\Gamma\vdash_{\mathcal S}t:A}^{!}_{a}$ is Lipschitz;

\item For $\mathcal S=\STLC,\RM{PCF}^{\Lawv}$, then $\model{\Gamma\vdash_{\mathcal S}M:A}^{!}_{a}$ is locally Lipschitz;

\item For $\mathcal S=\STLC$, $\Te{M}$ decomposes $\model{\Gamma \vdash_{\STLC} M:A}^!_{a}$ as an $\inf_{t\in\Te{M}}\model{\Gamma\vdash_{\STDLC} t:A}^!_{a}$ of {Lipschitz} functions.
\end{enumerate}

\end{theorem}
\begin{proof}
1). From Proposition~\ref{prop:troplinear} 2.~and the remark that for any type of $\BSTLC$, $\model{A}$ is finite.
2.) From Proposition~\ref{prop:troplinear} 2.~observing that a resource term $t(x)$ may use a variable $x$ a fixed number $n$ times, so that its matrix lies in $\Lawv^{!_{n}X\times Y}$. 
3). From Theorem~\ref{thmTLSlocLip}.
4). It follows from \autoref{cor:T(M)=M} plus the fact that, for $(f_n)_{n\in\N}\subseteq\Lawv^{!X\times Y}$, we have 
$\left(\inf_{n\in\N} f_n\right)^!=\inf_{n\in\N} f_n^!$.
\end{proof}

%
%
%

We conclude our discussion with an application of the Taylor expansion in $\LREL_{!}$: as proved in the previous section, all tps are locally Lipschitz; now, Theorem \ref{thm:taylor} can be used to compute approximations of the Lipschitz constants of an actual higher-order program.

\begin{corollary}\label{cor:taylor}
Suppose $x: A \vdash_{\STLC}M:B$ and $\vdash_{\STLC} N:A$. 
Then for all $t\in \Te{M}$ such that $\model{t}^{!}(\model{N})\neq \infty$, and $\delta>0$, the tps $\model{ x:A\vdash_{\STLC}M:B}^{!}$ is $\frac{\model{t}^{!}(\model{N}+2\delta)}{\delta}$-Lipschitz over the open ball
 $B_{\delta}(\model{N})$.
\end{corollary}
\begin{proof}
Thm.~\ref{thmTLSlocLip} yields the estimate $\max_{\overline{B_{3\delta}(\model{N})}}\model{M}^{!}$. As from Thm.~\ref{thm:taylor} 4.~it follows that $\model{t}^{!}\geq \model{M}^{!}$, we deduce that $K=\max_{\overline{B_{3\delta}(\model{N})}}\model{t}^{!}\geq \max_{\overline{B_{3\delta}(\model{N})}}\model{M}^{!}$ is also a local Lipschitz constant for $\model{M}^{!}$. Moreover, since $\model{t}^{!}$ is concave and non-decreasing, the $\max$ of $\model{t}^{!}$ is attained at the maximum point of $\overline{B_{3\delta}(\model{N})}$, that is, 
$K= \model{t}^{!}(x+3\delta)$. Finally, from $\model{t}^{!}(\model{N})<\infty$ and the continuity of $\model{t}^{!}$ we deduce $K<\infty$.
\end{proof}

\begin{example}
Consider again the term $M=zxx$ from Example \ref{ex:zxx}. 
%
 The (generalized) tps $\model{M}^{!}(x)(y)= \inf_{n,n'\in \mathbb N}\{y_{[(n,n')]}+(n+n')x\}$ is not (globally) Lipschitz: for any 
$L>0$, choose a natural number $N>L$, let $Y\in \Lawv^{\mathcal M_{\mathrm{fin}}(\mathbb N\times \mathbb N)}$ be such that $Y_{\mu}<\infty$ only if $\mu=[(n,n')]$ with $n+n'\geq N$; then $|\model{M}^{!}(x)(Y)- \model{M}^{!}(x+\epsilon)(Y)|\geq N\epsilon > L\epsilon$. 
Now take the approximant $t= z \langle x^{N-1}\rangle \langle x\rangle \in \Te{M}$ (chosen so that $\model t^{!}(x)(Y)<\infty$). Its interpretation is the monomial 
$\model{t}^{!}(x)(Y) = Y_{[(N-1,1)]}+Nx$. We can then compute a Lipschitz-constant for $\model{M}^{!}$ around $\langle x,Y\rangle$
as $\frac{1}{\delta}\model{t}^{!}(\langle x,Y\rangle+\delta)= 3N+3 + \frac{Y_{[(N-1,1)]}+Nx}{\delta}$. 
\end{example}

\section{Generalized Metric Spaces and $\Lawv$-Modules}\label{sec:GMS}

As we have seen, the morphisms of $\LREL$ can be seen as continuous functions between the $\Lawv$-modules $\Lawv^{X}$, when the latter are taken with the metric induced by the $\infty$-norm. This viewpoint gives a metric flavor to $\LREL$, and allowed us to relate differential and metric structure. Yet, how far can this correspondence be pushed?
In particular, is this correspondence restricted to $\Lawv$-modules of the form $\Lawv^{X}$ (i.e.~with a fixed base), or does it hold in some sense for arbitrary $\Lawv$-modules? Is this correspondence restricted to the $\infty$-norm metric, or does it hold for other metrics too?

\subsection{$\Lawv$-Modules and Cocomplete $\Lawv$-Categories}

An answer to the questions above comes from an elegant categorical correspondence between tropical linear algebra and the theory of \emph{generalized metric spaces}, initiated by Lawvere's pioneering work \cite{Lawvere1973}, and at the heart of the emergent field of \emph{monoidal topology} \cite{Hofmann2014, Stubbe2014}.

%
%
%


On the one hand we have $\Lawv$-modules: these are triples $(M,\preceq, \star)$ where $(M, \preceq)$ is a sup-lattice, and $\star: \Lawv \times M \to M$ is a continuous (left-)action of $\Lawv$ on it, where continuous means that $\star$ commutes with both joins in $\Lawv$ and in $M$. 
A $\Lawv$-module homomorphism is a map $f:M\to N$ commuting with both joins and the $\Lawv$-action. We let $\Mod$ indicate the category of $\Lawv$-modules and their homomorphisms. 

On the other hand we have Lawvere's \emph{generalized metric spaces}  \cite{Lawvere1973, Hofmann2014, Stubbe2014}:
Lawvere was the first to observe that a metric space can be described as a \emph{$\Lawv$-enriched} category. Indeed, spelling out the definition, a $\Lawv$-enriched category (in short, a $\Lawv$-category) is given by a set $X$ together with a ``hom-set'' $X(-,-):X\times X\to \Lawv$ satisfying 
$0  \geq X(x,x)$ and $X(y,z)+X(x,y)\geq  X(x,z)$. 
This structure clearly generalizes the usual definition of metric spaces, which are indeed precisely the  
$\Lawv$-categories which are \emph{skeletal} (i.e.~$X(x,y)=0$ implies $x=y$) and \emph{symmetric} (i.e.~$X(x,y)=X^{\mathrm{op}}(x,y)$, where $X^{\mathrm{op}}(x,y)=X(y,x)$).
A basic example of $\Lawv$-category is $\Lawv$ itself, with the distance $\Lawv(x,y)=y \menus x$. 

Moreover, a $\Lawv$-enriched \emph{functor} between $\Lawv$-categories is nothing but a non-expansive map $f:X\to Y$, since functoriality reads as $Y(f(x),f(y))\leq X(x,y)$.
Functors of shape $\Phi: X\times Y^{\op}\to \Lawv$ are called \emph{distributors} and usually noted $\Phi: Y \pfun X$. Notice that distributors $\Phi: Y\pfun X$ and $\Psi: Z\pfun Y$ can be composed just like 
ordinary matrices in $\LREL$: $\Psi\diamond \Phi : Z\pfun X$ is given by
$(\Psi\diamond \Phi)_{z,x}=\inf_{y\in Y}\Psi(z,y)+\Phi(y,x)$.

Lawvere also observed that the usual notion of Cauchy-completeness can be formulated, in this framework, as the existence of suitable colimits \cite{Lawvere1973}. Let us recall the notion of weighted colimit in this context:
 
\begin{definition}[weighted colimits]
Let $X,Y,Z$ be $\Lawv$-categories,
$\Phi: Z\pfun Y$ be a distributor and  $f:Y\to X$ be a functor.
A functor $g:Z\to X$ is the \emph{$\Phi$-weighted colimit of $f$ over $X$}, noted $\colim(\Phi,f)$, if for all $z\in Z$ and $x\in X$
\begin{align}
X(g(z),x)= \sup_{y\in Y}\left\{X(f(y),x)\menus \Phi(y,z)\right\}
\end{align} 

\end{definition}

A functor $f:X\to Y$ is said \emph{cocontinuous} if it commutes with all existing weighted colimits in $X$, i.e.~$f(\colim(\Phi,g))=\colim(\Phi,f\circ g)$. A $\Lawv$-enriched category.
A $\Lawv$-category $X$ is said \emph{cocomplete} if all weighted colimits over $X$ exist. 
We let $\GMet$ indicate the category of cocomplete $\Lawv$-categories and cocontinuous $\Lawv$-enriched functors as morphisms.

Observe that the usual Cauchy completeness for a $\Lawv$-category $X$ follows from its cocompleteness. Indeed, a Cauchy sequence $(x_{n})_{n\in \mathbb N}$ in $X$ yields two adjoint distributors $\phi^{*}:1\pfun X$ and $\phi_{*}:X\pfun 1$, where 
$\phi^{*}(x')=\lim_{n\to \infty}X(x_{n},x')$ and
$\phi_{*}(x')=\lim_{n\to \infty}X(x',x_{n})$. Hence,  
$\colim(\phi^{*},1_{X}):1\to X$ must be a point $x$ satisfying $0=X(x,x)=\sup_{y\in X}\lim_{n\to \infty}(X(y,x)\menus X(x_{n},y))$, which implies $\lim_{n\to \infty}X(x_{n},x)=0$. 


It turns out that the notions of $\Lawv$-module and cocomplete $\Lawv$-category are indeed equivalent. More precisely, the categories $\Mod$ and $\GMet$ are isomorphic \cite{Stubbe2006}. 
First, any $\Lawv$-module $(M,\preceq, \star)$ can be endowed with the structure of a $\Lawv$-category by letting
$M(x,y) = \inf\left\{ \epsilon \mid \epsilon \star x\geq y\right\}$. Moreover, a homomorphism of $\Lawv$-modules induces a cocontinuous functor of the associated $\Lawv$-categories. 
Conversely, in cocomplete $\Lawv$-categories it is possible to introduce a continuous $\Lawv$-action
via suitable weighted colimits called \emph{tensors} (cf.~\cite{Stubbe2014}):

\begin{definition}[tensors]
Let $X$ be a $\Lawv$-category, $x\in X$ and $\epsilon \in \Lawv$. The \emph{tensor of $x$ and $\epsilon$}, if it exists, is the colimit $\epsilon \otimes x:= \colim( [\epsilon],\Delta x)$, where
$[\epsilon]: \{\star\}\pfun \{\star\}$ is the constantly $\epsilon$ distributor
and $\Delta x:\{\star\}\to X$ is the constant functor. 
\end{definition}

A cocomplete $\Lawv$-category can thus be endowed with a $\Lawv$-module structure with order given by  $x\preceq_{X}y $ iff $X(y,x)=0$, and 
action given by tensors $\epsilon \otimes x$. 
Moreover, a cocontinuous functor between cocomplete $\Lawv$-categories is the same as a homomorphism of the associated $\Lawv$-modules. 

\subsection{Exponential and Differential Structure of $\Mod\simeq\GMet$}

We now show that the correspondence between $\Lawv$-modules and cocomplete $\Lawv$-categories lifts to a model of the differential $\lambda$-calculus, generalizing the co-Kleisli category $\LREL_{!}$. 

In order to define a Lafont exponential $!$ over $\Mod$, we exploit a well-known recipe from \cite{Mellies2018, Manzo2013}.
The first step is to define a symmetric algebra $\Sym_{n}(M)$ as the equalizer of all permutative actions on $n$-tensors $M\otimes \dots \otimes M$.
Notice that each element of $!_{n}M$ can be described as a join of ``multisets''
$[x_{1},\dots, x_{n}]$, where the latter is the equivalence class of the tensors 
$x_{1}\otimes \dots \otimes x_{n}\in M^{\otimes_{n}}$ under the action of permutations $\sigma\in \F S_{n}$.
The $\Lawv$-module $!_{n}M$ is a cocomplete $\Lawv$-category with distance function defined on basic ``multisets'' as follows:
\begin{align}
(!_{n}M)(\alpha,\beta)=
\sup_{\sigma\in \F S_{n}}\inf_{\tau\in \F S_{n}}\sum_{i=1}^{n}
X(x_{\sigma(i)},y_{\tau(i)})
\end{align}
where $\alpha=[x_{1},\dots, x_{n}]$ and $\beta= [y_{1},\dots, y_{n}]$, and extended to arbitrary elements $\alpha=\bigvee_{i}\alpha_{i}$ and $\beta=\bigvee_{j}\beta_{j}$
by $(!_{n}M)(\alpha,\beta)=\sup_{i}\inf_{j}(!_{n}M)(\alpha_{i},\beta_{j})$.

Next, we define $!M$ as the infinite biproduct $\prod_{n}!_{n}M$, yielding the cofree commutative comonoid over $M$ (cf.~\cite[Proposition 1]{Mellies2018}).
Using the fact that biproducts commute with tensors in $\LREL$, by standard results \cite{Mellies2018}, we obtain that the coKleisli category $\Mod_{!}$ is a CCC.
Moreover, the constructions for $\Mod$ generalize those of $\LREL$, in the sense that $!(\Lawv^{X})\simeq \Lawv^{\multiset(X)}$ and that $\Mod_{!}(\Lawv^{X},\Lawv^{Y})\simeq \LREL_{!}(X,Y)$.

Finally, since the coKleisli category of a Lafont category with biproducts is always a CC$\partial$C \cite[Theorem 21]{LemayCALCO2021}, we can endow $\Mod_{!}$ with a  differential operator $E$,  generalizing $D^{!}$, given by 
\begin{align}\label{eq:dermod}
Ef(\alpha)=
\bigvee\left\{
f(\beta\cup [x]) \ \Big \vert  \ 
\iota_{n}(\beta)\otimes \iota_{1}(x) \leq S(\alpha)
\right\}
\end{align}
where $\iota_{k}: M_{k}\to \prod_{i\in I}M_{i}$ is the injection morphism given by $\iota_{k}(x)( k)=x$ and $\iota_{k}(x)(i\neq k)=\infty$,
and $S: !(M\times N)\to !M\otimes !N$ is the Seely isomorphism \cite{Mellies2018}, and conclude that:
\begin{theorem}\label{thm:lemay}
($\Mod_{!}/\GMet_{!},E$) is a CC$\partial$C.
\end{theorem}

\section{Related Work}\label{section7} 

The applications of tropical mathematics in computer science abound, e.g.~in automata theory \cite{Chua1992, Simon}, machine learning \cite{Maragos2021, Pachter2004, Zhang2018}, optimization \cite{Akian2011, Akian2012}, and convex analysis \cite{Lucet2009}. 

As we said, the relational semantics over the tropical semiring was quickly explored in \cite{Manzo2013}, to provide a ``best case'' resource analysis of a $\mathrm{PCF}$-like language with non-deterministic choice. 
The connections between differential $\lambda$-calculus (and differential linear logic), relational semantics, and non-idempotent intersection types are very well-studied (see \cite{decarvalho2018}, and more recently, \cite{Mazza2016} for a more abstract perspective, and \cite{Olimpieri2021, Galal2021} for a 2-categorical, or proof-relevant, extension).
\emph{Probabilistic coherent spaces} \cite{Ehrhard2011}, a variant of  the relational semantics, provide an interpretation of higher-order probabilistic programs
as analytic functions. In \cite{Ehrhard2022} it was observed that such functions satisfy a local Lipschitz condition somehow reminiscent of our examples in Section \ref{section2}.

The study of linear or bounded type systems for sensitivity analysis was initiated in \cite{Girard92tcs} and later developed \cite{Schopp, SchoppDalLago, Reed2010}.
Related approaches, although not based on metrics, are provided by \emph{differential logical relations} \cite{dallago} and \emph{change action} models \cite{Picallo2019}.
More generally, the literature on program metrics in denotational semantics is vast. Since at least \cite{VANBREUGEL20011} metric spaces, also in Lawvere's generalized sense \cite{Lawvere1973}, have been exploited as an alternative framework to standard, domain-theoretic, denotational semantics; notably models of $\STLC$ and $\mathrm{PCF}$ based on 
\emph{ultra}-metrics and \emph{partial} metrics have been proposed  \cite{Escardo1999,PistoneLICS, PistoneFSCD2022}.

Motivated by connections with computer science and fuzzy set-theory, 
the abstract study of generalized metric spaces in the framework of \emph{quantale}- or even \emph{quantaloid}-enriched categories has led to a significant literature in recent years (e.g.~\cite{Hofmann2014, Stubbe2014}), 
and connections with tropical mathematics also have been explored e.g.~in \cite{Fuji, Willerton2013}. Moreover, applications of quantale-modules to both logic and computer science have also been studied \cite{Abramsky1993b, Russo2007}.

Finally, connections between program metrics and the differential $\lambda$-calculus have been already suggested in \cite{PistoneLICS}; moreover, \emph{cartesian difference categories} \cite{Picallo2020} have been proposed as a way to relate derivatives in differential categories with those found in change action models.

%
%
%
%

\section{Conclusion and Future Work}\label{section8} 
%
%
%

The main goals of this paper are two. Firstly,  to
demonstrate the existence of a conceptual bridge between two well-studied quantitative approaches to higher-order programs, and to highlight the possibility of transferring results and techniques from one approach to the other. 
Secondly, to suggest that tropical mathematics, a
field which has been largely and successfully applied in computer science, could be used for the quantitative analysis of functional programming languages. 
While the first goal was here developed in detail, and at different levels of abstraction, for the second goal we only sketched a few interesting directions, and we leave their development to a second paper of this series. 

While the main ideas of this article only use basic concepts from the toolbox of tropical mathematics, an exciting direction is that of looking at potential applications of more advanced tools from tropical algebraic and differential geometry (e.g.~Newton polytopes, tropical varieties, tropical differential equations). Another interesting question is how much of our results on tps and their tropical Taylor expansion can be extended to the abstract setting of generalized metric spaces and continuous functors. 


\bibliography{main}

\appendix

\section{Proofs from \autoref{section2}: Theorem \ref{theorem:fepsilon}}

We give below the complete statement of Theorem \ref{theorem:fepsilon} together with its proof.

First, let us set the following:
\begin{definition}
 Let $\preceq$ be the product order on $\N^k$ (i.e.\ for all $ m  , n  \in \N^{K}$, $ m  \preceq  n  $ iff $m_{i}\leq n_{i}$ for all $1\leq i\leq K$).
 Of course $ m  \prec  n  $ holds exactly when $ m  \preceq  n  $ and $m_{i}<n_{i}$ for at least one $1\leq i\leq K$.
 Finally, we set $ m  \prec_{1} n  $  iff
$ m  \prec  n  $ and $\sum_{i=1}^{K}n_{i}-m_{i}=1$ (i.e.\ they differ on exactly one coordinate).
\end{definition}

We will exploit the following:

\begin{remark}\label{rmk:AC}
\text{If $U\subseteq \N^{K}$ is infinite, then $U$ contains an infinite ascending chain $ m  _{0}\prec  m  _{1} \prec  m  _{2} \prec \dots$.}.

This is a consequence of K\"onig Lemma (KL): consider the directed acyclic graph $(U,\prec_{1})$, indeed a $K$-branching tree; if there is no infinite ascending chain $  m  _{0}\prec  m  _{1} \prec  m  _{2} \prec \dots$, then in particular there is no infinite ascending chain $  m  _{0}\prec_{1}  m  _{1} \prec_{1}  m  _{2} \prec_{1} \dots$ so the tree $U$ has no infinite ascending chain; then by KL it is finite, contradicting the assumption. 
\end{remark}

\begin{theorem}[Theorem \ref{theorem:fepsilon}]
Let $k\in\N$ and $f:\Lawv^k\to\Lawv$ a tps with matrix $\matr f:\N^k\to\Lawv$.
 For all $0<\epsilon<\infty$, there is $\C F_\epsilon \subseteq \N^k$ such that:
\begin{enumerate}
 \item $\C F_\epsilon$ is finite
 \item If $\mathcal{F}_\epsilon= \emptyset$ then $f( x ) = +\infty$ for all $ x \in \Lawv^k$
 \item If $f( x _0) = +\infty$ for some $ x _0\in [\epsilon,\infty)^{K}$ then $\mathcal{F}_\epsilon= \emptyset$
 \item The restriction of $f$ on $[\epsilon,\infty]^k$ coincides  with the tropical {polynomial} \[P_\epsilon(x):=\min\limits_{n\in \C F_\epsilon}\set{nx+\matr f(n)}\]
where $nx:=\sum_{i=1}^k n_ix_i$.
\end{enumerate}
\end{theorem}
\begin{proof}
We let $\mathcal F_\epsilon$ to be the complementary in $\N$ of the set:
\[
 \set{ n  \in\N^{K} \mid \textit{either } \hat f ( n  )=+\infty \textit{ or there is }  m  \prec  n  \textit{ s.t.\ } \hat f( m  )\leq\hat f( n  )+\epsilon}.
\]
In other words, $ n  \in\mathcal F_\epsilon$ iff $\hat f( n  )<+\infty$ and for all $ m  \prec  n  $, one has $\hat f( m  )>\hat f( n  )+\epsilon$.

1).
Suppose that $\mathcal F_\epsilon$ is infinite; then, using Remark~\ref{rmk:AC}, it contains an infinite ascending chain
\[\set{ m  _0\prec  m  _1\prec\cdots}.\]
By definition of $\mathcal F_\epsilon$ we have then:
\[+\infty>\hat f( m  _0)>\hat f( m  _1)+\epsilon>\hat f( m  _2)+2\epsilon>\cdots\]
so that $+\infty>\hat f( m  _0)>\hat f( m  _{i})+i\epsilon\geq i\epsilon$ for all $i\in\N$.
This contradicts the Archimedean property of $\R$.

2).
We show that if $\mathcal F_\epsilon=\emptyset$, then $\hat f( n  )=+\infty$ for all $ n  \in\N^{K}$.
This immediately entails the desired result.
We go by induction on the well-founded order $\prec$ over $ n  \in\N^{K}$:

- if $ n  =0^{K}\notin\mathcal F_\epsilon$, then $\hat f( n  )=+\infty$, because there is no $ m  \prec n  $.

- if $ n  \notin\mathcal F_\epsilon$, with $ n  \neq 0^{K}$ then either $\hat f( n  )=+\infty$ and we are done, or there is $ m  \prec  n  $ s.t.\ $\hat f( m  )\leq \hat f( n  )+\epsilon$.
By induction $\hat f( m  )=+\infty$ and, since $\epsilon<+\infty$, this entails $\hat f( n  )=+\infty$.

3).
If $f( x _0)=+\infty$ with $ x _0\in [\epsilon,\infty)^{K}$, then necessarily $\hat f( n  )=+\infty$ for all $ n  \in\N^{K}$.
Therefore, no $ n  \in\N^{K}$ belongs to $\mathcal F_\epsilon$.

4).
We have to show that $f( x )=P_\epsilon( x )$ for all $ x \in [\epsilon,+\infty]^{K}$.
By 1), it suffices to show that we can compute $f( x )$ by taking the $\inf$, that is therefore a $\min$, only in $\mathcal F_\epsilon$ (instead of all $\N^{K}$).
If $\mathcal F_\epsilon=\emptyset$ then by 2) we are done (remember that $\min\emptyset := +\infty$).
If $\mathcal F_\epsilon\neq\emptyset$, we show that for all $ n  \in\N^{K}$, if $ n   \notin\mathcal F_\epsilon$, then there is $ m  \in\mathcal F_\epsilon$ s.t.\ $\hat f( m  )+ m   x  \leq \hat f( n  )+ n   x $.
We do it again by induction on $\prec_{1}$:

- if $ n  =0^{K}$, then from $\mathbf  n\notin \mathcal F_{\epsilon}$, by definition of $\mathcal F_\epsilon$, we have $\hat f( n  )=+\infty$ (because there is no $ n  '\prec n  $).
So any element of $\mathcal F_\epsilon\neq\emptyset$ works.

- if $ n  \neq 0^{K}$, then we have two cases:
either $\hat f( n  )=+\infty$, in which case we are done as before by taking any element of $\mathcal F_\epsilon\neq\emptyset$.
Or $\hat f( n  )<+\infty$, in which case (again by definition of $\mathcal F_\epsilon$) there is $ n  '\prec n  $ s.t.\ \begin{equation}\label{eq:n'neps} \hat f( n  ')\leq \hat f( n  )+\epsilon.\end{equation}
Therefore we have (remark that the following inequalities hold also for the case $x=+\infty$):
\[\begin{array}{rclr}
 \hat f( n  ')+ n  ' x  & \leq & \hat f( n  ) + \epsilon +  n' x  & \textit{by \eqref{eq:n'neps}} \\
 & \leq & \hat f( n  ) + ( n  - n  ') x  +  n  ' x  & \textit{because $\epsilon\leq\min x $ and $\min  x \leq( n  -  n') x $} \\
 & = & \hat f( n  )+  n   x . &
\end{array}\]
Now, if $ n  '\in\mathcal F_\epsilon$ we are done.
Otherwise $ n  '\notin\mathcal F_\epsilon$ and we can apply the induction hypothesis on it, obtaining an $ m  \in\mathcal F_\epsilon$ s.t.\ $\hat f( m  )+ m   x  \leq \hat f( n  ')+ n  ' x $.
Therefore this $ m  $ works.
\end{proof}

\section{Proofs from \autoref{section3}: $(\LREL, !_{n})$ is a model of $\BSTLC$}


Given SMCs $\C C,\C D$,  let $\B{SMC}_{l}(\C C, \C D)$ indicate the category of symmetric lax monoidal functors and monoidal natural transformations between them.
$\B{SMC}_{l}(\C C, \C D)$  is itself a SMC, with monoidal structure defined pointwise.

The set 
$\BB N$ can be seen as the category with objects the natural numbers and a morphism between $r$ and $r'$ precisely when $r\leq r'$. 
Moreover, $\BB N$ can be seen as a SMC in two ways:
\begin{itemize}

\item we indicate as $\BB N^{+}$ the SMC with monoidal product given by addition;
\item we indicate as $\BB N^{\times}$ the SMC with monoidal product given by multiplication.
\end{itemize}


\begin{definition}[cf.~\cite{Katsumata2018}]
A \emph{$\BB N$-graded linear exponential comonad} on a symmetric monoidal category $\C C$ is a tuple
$(D, w,c,\epsilon,\delta)$ where:
\begin{itemize}

\item $D: \BB N\to \B{SMC}_{l}(\C C, \C C)$ is a functor. We write 
$m_{r}:\{\star\} \to D(r)(\{\star\})$ and $m_{r,A,B}: D(r)(A)\otimes D(r)(B) \to D(r)(A\otimes B)$ for the symmetric lax monoidal structure of $D(r)$;

\item $(D,w,c): \BB N^{+}\to \B{SMC}_{l}(\C C, \C C)$ is a symmetric colax monoidal functor;

\item $(D,\epsilon,\delta):\BB N^{\times}\to (\B{SMC}_{l},\mathrm{Id},\circ) $ is a colax monoidal functor.

\end{itemize}
further satisfying the axioms below:
\begin{align}
w_{A}& =  w_{D(s)(A)}\circ \delta_{0,s,A}\\
w_{A} & = D(s)(w_{A} )\circ \delta_{s,0,A} \\
(\delta_{r,s,A}\otimes \delta_{r',s,A})\circ c_{rs,r's,A}
&=
c_{r,r',D(s)(A)}\circ \delta_{r+r',s,A}\\
m_{s,D(r)(A),D(r')(A)}\circ (\delta_{r,s,A}\otimes \delta_{s,r',A})\circ c_{sr,sr',A}&=
D(s)(c_{r,r',A})\circ \delta_{s,r+r',A}
\end{align}
\end{definition}

Concretely, the definition above requires 6 natural transformations:
\begin{align*}
m_{r} & :\{\star\}\to  D(r)(\{\star\})\\
m_{r,A,B}& :  D(r)(A)\otimes D(r)(B)\to  D(r)(A\otimes B)\\
w_{A}& : D(0)(A)\to \{\star\} \\
c_{r,r',A}& : D(r+r')(A) \to D(r)(A)\otimes D(r')(A) \\
\epsilon_{A}& : D(1)(A)\to A \\
\delta_{r,r',A}&: D(r r')(A)\to D(r)(D(r')(A))
\end{align*}
subject to the following list of equations:
\begin{itemize}
\item $D(r)$ is a lax monoidal functor:
\begin{align}
m_{r,A\otimes B,C}\circ (m_{r,A,B}\otimes D(r)(C)) & = 
m_{r,A, B\otimes C}\circ (D(r)(A)\otimes m_{r,B,C})\\
m_{r,A,\{\star\}}\circ (D(r)(A)\otimes m_{r}) & = D(r)(A)\\
m_{r,\{\star\}, B}\circ (m_{r}\otimes D(r)(B))&= D(r)(B)
\end{align}

\item $(D,w,c)$ is a symmetric colax monoidal functor:
\begin{align}
(c_{r,s,-}\otimes D(t)(-))\circ c_{r+s,t} & =
(D(r)(-)\otimes c_{s,t,-})\circ c_{r,s+t}\\
(D(r)(-)\otimes w_{-})\circ c_{r,0,-} & = D(r)(-) \\
(w_{-}\otimes D(r)(-))\circ c_{0,r,-} & = D(r)(-)
\end{align}

\item $(D,\epsilon,\delta)$ is a colax monoidal functor:
\begin{align}
\delta_{r,s, D(t)(-)}\circ \delta_{(rs),t,-} & =
D(r)(\delta_{s,t,-})\circ \delta_{r,st,-}\\
D(r)(\epsilon_{-}) \circ \delta_{r,1,-} & = D(r)(-) \\
\epsilon_{D(r)(-)} \circ \delta_{1,r,-} & = D(r)(-)
\end{align}

\end{itemize}

The following definition provides an interpretation of $\BSTLC$ in any symmetric monoidal closed category with a $\BB N$-graded linear exponential comonad.

\begin{definition}[interpretation of $\BSTLC$]
Let $\C C$ be a symmetric monoidal closed category and $(D, w,c,\epsilon,\delta)$ be a $\BB N$-graded linear exponential comonad. 
Let $\model{X}$ be fixed objects of $\C C $, one for each ground type $X$ of $\BSTLC$. 

One lifts the interpretation to types as 
$\model{!_{n}A\multimap B}=D(n)(\model{A})\multimap \model B$. Moreover, one extends the interpretation to contexts via 
$\model{\{x_{1}:_{n_{1}}A_{1},\dots, x_{k}:_{n_{k}}A_{k}\}}=
\bigotimes_{i=1}^{k} D(n_i)(\model A_i)$.

Then, one inductively defines an interpretation $\model{\Gamma \vdash M:A}\in \C C(\model\Gamma, \model A)$ of $\Gamma \vdash M:A$ by induction as follows:
 \begin{itemize}
\item $\model{x:_{1}A\vdash x: A}=\epsilon_{A}$;
\item $\model{\Gamma,x:_{0}B\vdash M:A} =
\model{\Gamma \vdash M:A}
\circ (\model{\Gamma}\otimes  w_{\model B})
$;
\item $\model{\Gamma, x:_{m+n}B \vdash M[x/y]:A}=
\model{\Gamma, x:_{m}B, y:_{m}B \vdash M:A}
\circ
(\model\Gamma \otimes c_{m,n,\model B})
$;
\item $\model{\Gamma \vdash \lambda x.M: !_{n}A\multimap B}=
\Lambda (\model{\Gamma, x:_{n}A \vdash M:B})
$, where $\Lambda$ is the isomorphism $ \C C(\model \Gamma \otimes D(n)(\model A), \model B) \to \C C(\model \Gamma, D(n)(\model A)\multimap \model B)$;

\item $\model{\Gamma+\Delta \vdash MN:B}=
\mathsf{ev}\circ \big( \model{\Gamma \vdash M:A\multimap B}
\otimes
\model{\Delta \vdash N:A}\big)$, where $\mathsf{ev}\in \C C((\model{A}\multimap \model B)\otimes \model A, \model B)$ is the evaluation morphism of $\C C$;

\item $\model{n\Gamma\vdash M:!_{n}A}=
 !_{n}(\model{\Gamma \vdash M:A})\circ \big(\delta_{n,m_{1},\model{A_{1}}}\otimes \dots \otimes 
\delta_{n,m_{k},\model{A_{k}}}\big)
$, where $\Gamma=\{x_{1}:_{m_{1}}A_{1},\dots, x_{k}:_{m_{k}}A_{k}\}$.

\end{itemize}
\end{definition}

Let us now show that bounded multisets defined a $\BB N$-graded linear exponential comonad over $\LREL$.

\begin{definition}
We define the following structure $(!_{-}(-),w,c,\epsilon,\delta)$ over the category $\LREL$ as follows:
\begin{itemize}
\item for any set $X$ and $n\in \BB N$, let $!_{n}(X)=\C M_{\leq n}(X)$;

\item for all $f: X\times Y\to \Lawv$, let $!_{n}(f): !_{n}(X)\times !_{n}(Y)\to \Lawv$ be defined by 
\begin{align*}
!_{n}(f)(\alpha,\beta)=
\begin{cases}
\min_{\sigma\in \F S_{k}}\sum_{i=1}^{k}f(x_{i},y_{\sigma(i)}) & 
\text{ if }\alpha=[x_{1},\dots, x_{k}], \beta=[y_{1},\dots, y_{k}]\\
\infty & \text{ otherwise}
\end{cases}
\end{align*}

\item $m_{r}(\star, \{\star\})=0$ and $m_{r}(\star, \emptyset)=\infty$;

\item $m_{r,A,B}: D(r)(A)\times D(r)(B)\times D(r)(A\times B)\to \Lawv$ is defined by 
\begin{align*}
m_{r,A,B}((\alpha,\beta), \gamma)=
\begin{cases}
0 & \text{ if } \alpha=[x_{1},\dots, x_{k}], \beta=[y_{1},\dots, y_{k}], \gamma= [(x_{1},y_{1}),\dots, (x_{k},y_{k})]\\
\infty & \text{ otherwise}
\end{cases}
\end{align*}

\item $w_{A}:D(0)(A)\times \{\star\}\to \Lawv$ is given by $w_{A}(\emptyset, \star)=0$ and is $\infty$ otherwise (observe that $D(0)(A)\simeq \{\star\}$);

\item $c_{r,s,A}: D(r+s)(A)\times D(r)(A)\times D(s)(A)\to \Lawv$ is given by $c_{r,r',A}(\langle\alpha, \beta,\gamma\rangle)=0$ if $\alpha=\beta+\gamma$, and is $\infty$ otherwise;

\item $\epsilon_{A}(\emptyset, a)=\infty$, $\epsilon_{A}([a],a)=0$, $\epsilon_{A}([b],a)=\infty$ $(b\neq a)$,

\item $\delta_{r,r',A}(\alpha, B)=0$ if $\alpha= \sum B$ (where $\sum B$ indicates the multiset obtained by the sum of all multisets contained in $B$) and is $\infty$ otherwise.

\end{itemize}

\end{definition}

\begin{proposition}
 $(!_{-}(-),w,c,\epsilon,\delta)$  is a $\BB N$-graded linear exponential comonad over $\LREL$.
\end{proposition}
\begin{proof}
\begin{itemize}

\item $D(r)$ is a lax monoidal functor:
 $$ m_{r,A\times B,C}\circ (m_{r,A,B}\times D(r)(C))(\langle \alpha,\beta,\gamma,\delta\rangle)
 :
 D(r)(A)\times D(r)(B)\times D(r)(C) \times D(r)(A\times B\times C)\to \Lawv
 $$
 is equal to $0$ 
precisely when $\alpha=[x_{1},\dots, x_{k}]$, $\beta=[y_{1},\dots, y_{k}]$, $\gamma=[z_{1},\dots, z_{k}]$ and 
$\delta= [(x_{1},y_{1},z_{1}),\dots, (x_{k},y_{k},z_{k})]$, and is $\infty $ in all other cases.

Observe that
$m_{r,A,B\times C}\circ (D(r)(A)\times m_{r,B,C})(\langle\alpha,\beta,\gamma, \delta\rangle)
  $ is equal to $0$ in the same situation, and is $\infty$ otherwise.
 
 We conclude that the two matrices coincide.
 
 Furthermore, we have that 
 $m_{r,A,\{\star\}}\circ (D(r)(A)\times m_{r})(\langle \alpha,\beta\rangle): D(r)(A)\times \{\star\} \times D(r)(A)$ is equal to $0$ 
 precisely when $\alpha=\beta$ and is $\infty$ otherwise, that is, it coincides with $\mathrm{id}_{D(r)(A)}$.

\item $(D,w,c)$ is a symmetric colax monoidal functor.

$((c_{r,s,A}\times D(t)(A))\circ c_{r+s,t,A}) (\langle \alpha,\beta,\gamma,\delta\rangle)
: D(r+s+t)(A)\times D(r)(A)\times D(s)(A)\times D(t)(A)$
is equal to $0$ when $\alpha=\beta+\gamma+\delta$, and is $\infty$ otherwise, and the same holds for
$((D(r)(A)\times c_{s,t,A})\circ c_{r,s+t,A}) (\langle \alpha,\beta,\gamma,\delta\rangle)
$.

Furthermore, 
$((D(r)(A)\times w_{A})\circ c_{r,0,A})(\alpha,\beta)
: D(r)(A)\times D(r)(A)\to \Lawv$ is equal to $0$ when 
$\alpha=\beta$, and is $\infty$ otherwise, so it coincides with 
$\mathrm{id}_{D(r)(A)}$.

\item $(D,\epsilon,\delta)$ is a colax monoidal functor:

$(\delta_{r,s,D(t)(A)}\circ \delta_{rs,t,A})
(\alpha, \Gamma)
: D(rst)(A) \times D(r)(D(s)(D(t)(A))) \to \Lawv
$
is $0$ precisely when $\alpha = \sum \sum \Gamma$, and is $\infty$ otherwise, and similarly for 
$(D(r)(\delta_{s,t,A})\circ \delta_{r,st,A})(\alpha,\Gamma)$.

Furthermore, 
$(D(r)(\epsilon_{A})\circ \delta_{r,1})( \alpha,\beta ):
D(r)(A) \times D(r)(A)\to \Lawv
$ is equal to $0$ when $\alpha=\beta$ and is $\infty$ otherwise, so it coincides with $\mathrm{id}_{D(r)(A)}$.

\end{itemize}

Let us check the further equations:
\begin{itemize}

\item $(w_{D(s)(A)}\circ \delta_{0,s,A})(\langle \emptyset,\star\rangle: D(0)(A)\times \{\star\}\to \Lawv$ is 0, precisely like $w_{A}$.

\item A similar argument holds for the second equation.

\item $((\delta_{r,s,A}\times \delta_{r',s,a})\circ c_{rs,r's,A})
(\langle \alpha, \Gamma,\Delta  \rangle)
:
D(rs+r's)(A)\times  D(r)(s)(A)\times D(r')(s)(A)\to \Lawv
$
is equal to $0$ when $\alpha=\sum \Gamma + \sum \Delta$, and is $\infty$ otherwise.

Now, 
using the fact that $D(rs+r's)(A)=D((r+r')s)(A)$, we can check that the same holds for 
$c_{r,r',D(s)(A)}\circ \delta_{r+r',s,A})(\langle \alpha, \Gamma,\Delta  \rangle)$: it is $0$ when 
$\alpha= \sum\Gamma+\Delta= \sum \Gamma+\sum \Delta$.

\item A similar argument holds for the fourth equation.

\end{itemize}
\end{proof}

\section{Proofs from \autoref{sec:tls}: the Analysis of Tropical Power Series}\label{appsec:tls}
\subsection{Proof of \autoref{prop:nondecr+conc}}

Remember that a function $f:Q^{X}\to Q^{Y}$ is \emph{concave} if for all $\alpha\in [0,1]$, $ x ,  y  \in Q^{X}$ and $b\in Y$ 
\[
f(\alpha\cdot  x +(1-\alpha)\cdot  y  )_{b} \geq \alpha f( x )_{b} + (1-\alpha)f( y  )_{b}
\]

We now prove \autoref{prop:nondecr+conc}, i.e.\ that $f: Q^{X}\to Q^{Y}$ are non-decreasing and concave.

The fact that $f$ is non-decreasing is clear, since the multiplicities of the multisets and all coordinates of the points are non-negative.
Let us show the concavity.
Let us first show that all functions of the form $f( x )_{b}= \mu  x + c$ are concave:
we have $f(\alpha x + (1-\alpha) y  )_{b}= \mu(\alpha x )+(1-\alpha) y  )+c=
 \mu(\alpha x )+(1-\alpha) y  )+\alpha c+(1-\alpha)c=
 \alpha(\mu  x  + c)+(1-\alpha)(\mu  y  +c)=\alpha f( x )_{b}+(1-\alpha) f( x )_{b}$.
To conclude, let us show that if $(f_{i})_{i\in I}$ is a family of concave functions from $Q^{X}$ to $Q^{Y}$, the function $f=\inf_{i\in I}f_{i}$ is also concave: we have
$f(\alpha x  +(1-\alpha) y  )_{b}=
\inf_{i\in I}f_{i}(\alpha x +(1-\alpha) y  )_{b} \geq 
\inf_{i\in I}\alpha f_{i}( x )_{b}+(1-\alpha)f_{i}( y  )_{b}
\geq 
\inf_{i\in I}\alpha f_{i}( x )_{b} + \inf_{j\in I}(1-\alpha)f_{j}( y  )_{b}
=
\alpha \cdot (\inf_{i\in I}f_{i}( x )_{b})+ (1-\alpha)\cdot( \inf_{j\in I}f_{j}( y  )_{b})=
\alpha  f( x )_{b}+(1-\alpha)f( y  )_{b}$, where we used the fact that given families $a_{i},b_{i}$ of reals,
$\inf_{i}a_{i}+b_{i}\geq \inf_{i}a_{i}+\inf_{j}b_{j}$.
This follows from the fact that for all $i\in I$, $a_{i}+b_{i}\geq \inf_{i}a_{i}+\inf_{i}b_{i}$.

\subsection{Proof of \autoref{thm:ScottCont}}

The part of \autoref{thm:ScottCont} about tPs immedley follows from the first part of the same theorem.
Let us quickly recall the basic definitions about cones that we need in order to prove it.

\begin{definition}
 An \emph{$\overline{\R}_{\geq 0}$-cone} is a commutative $\overline{\R}_{\geq 0}$-semimodule with cancellative addition (i.e.\ $x+y=x+y' \Rightarrow y=y'$).
\end{definition}

In \cite{Selinger2004} cones are required to also have ``strict addition'', meaning that $x+y=0 \Rightarrow x=y=0$.
We do not add this requirement since it will automatic hold when considering normed cones.

\begin{remark}
 The addition of a cone $P$ (which forms a commutative monoid) turns $P$ into a poset by setting:
 $
  x \leq y \textit{ iff } y=x+z \textit{, for some }z\in P.
 $
 This is called the \emph{cone-order} on $P$
 By the cancellative property, when such $z$ exists it is unique, and we denote it by $y-x$.
\end{remark}

\begin{definition}
 A \emph{normed $\overline{\R}_{\geq 0}$-cone} $P$ is the data of a $\overline{\R}_{\geq 0}$-cone together with a $\leq$-monotone\footnote{I.e.: $x\leq y \Rightarrow \norm{x}\leq \norm y$. Remark that requiring this property (for all $x,y$) is equivalent to requiring that $\norm{x}\leq \norm{x+y}$ for all $x,y$.} norm on it, where a \emph{norm} on $P$ is a map $\norm{.}:P\to \overline{\R}$ satisfying the usual axioms of norms:
 $\norm x \geq 0$, $\norm x = 0 \Rightarrow x=0$, $\norm{rx}=r\norm x$ and $\norm{x+y}\leq \norm x + \norm y$.
\end{definition}

In, e.g.~\cite{EhrPagTas2018}, a normed $\overline{\R}_{\geq 0}$-cone is simply called a cone.

Remark that in a normed $\overline{\R}_{\geq 0}$-cone, by monotonicity of the norm, we have: $\norm{x+y}=0 \Rightarrow x=y=0$.
Therefore, as already mentioned, in a normed cone we have: 
$x+y=0 \Rightarrow \norm{x+y}=0 \Rightarrow x=y=0$, that is, addition is strict.

\begin{example}
 $\overline{\R}_{\geq 0}^X$ is a normed cone with the norm $\supnorm{x}:=\sup\limits_{a\in X} x_a\in \overline{\R}_{\geq 0}$.
\end{example}

\begin{remark}
 The cone-order on $\overline{\R}_{\geq 0}^X$ is the pointwise usual order on $\overline{\R}_{\geq 0}$.
\end{remark}

A \emph{directed net} in a poset $P$ with indices in a set $I$ is a function $s:I\to P$, denoted by $(s_i)_{i\in I}$, s.t.\ its image is directed.
We say that a directed net in $P$ \emph{admits a sup} iff its image admits a sup in $P$.
We say that a directed net $s$ in a normed cone is \emph{bounded} iff the set $\set{\norm{s_i}\,\mid i\in I}$ is bounded in $\R_{\geq 0}$.

Remember the definition of Scott-continuity:

\begin{definition}
 A function $f:P\to P'$ between posets is \emph{Scott-continuous} iff for all directed net $(s_i)_i$ in $P$ admitting a sup, we have $\exists \bigvee\limits_i f(s_i) = f(\bigvee\limits_i s_i)$ in $P'$. 
\end{definition}

The fundamental result in order to prove Theorem~\ref{thm:ScottCont} is the following, taken from \cite{Selinger2004}.

\begin{proposition}\label{prop:infsup}
 Let $P$ be a normed $\overline{\R}_{\geq 0}$-cone s.t.\ every bounded directed net in $P$ admits a sup.
 Let $(v_i)_{i\in I}$ be a directed net in $P$ with an upper bound $v\in P$.
 Then $\exists\bigvee\limits_{i\in I} v_i \in P$ and, if $\inf\limits_{i\in I} \norm{v-v_i} =0$, one has: $\bigvee\limits_{i\in I} v_i = v$.
\end{proposition}
\begin{proof}
 Remark that $v-v_i$ exists in $P$ by hypothesis and so does $\bigvee\limits_{i\in I} v_i$. 
 Now, since $v\geq v_i$ for all $i$, we have that $v\geq \bigvee\limits_{i\in I} v_i$, and so $v-\bigvee\limits_{i\in I} v_i$ exists in $P$.
 Fix $i\in I$.
 Since $v_i\leq \bigvee\limits_{i\in I} v_i$, then $v-\bigvee\limits_{i\in I} v_i\leq v-v_i$ and, by monotonicity of the norm, $\norm{v-\bigvee\limits_{i\in I} v_i}\leq \norm{v-v_i}$.
 Since this holds for all $i\in I$, we have:
 $0\leq \norm{v-\bigvee\limits_{i\in I} v_i}\leq \inf\limits_{i\in I} \norm{v-v_i}=0$, where the last equality holds by hypothesis.
 Thus $\norm{v-\bigvee\limits_{i\in I} v_i}=0$, i.e.\ $v=\bigvee\limits_{i\in I} v_i$.
\end{proof}

We finally obtain the desired:

\begin{theorem}[Theorem~\ref{thm:ScottCont}]
  All monotone (w.r.t.\ pointwise order) and $\norm{\cdot}_{\infty}$-continuous functions $f:(0,\infty)^X\to (0,\infty)$ are Scott-continuous.
\end{theorem}
\begin{proof}
 Let $(x_i)_i$ a directed net in $(0,\infty)^X$ s.t.\ $\bigvee\limits_i x^i$ exists in $(0,\infty)^X$.
 Then $\inf\limits_i \supnorm{\bigvee\limits_i x^i - x^i} =0$, where $\bigvee\limits_i x^i - x^i$ exists because $\bigvee\limits_i x^i \geq x^i$ for all $i$.
 Since $f$ is $\supnorm{.}$-continuous on $(0,\infty)^X$, then $\inf\limits_i \supnorm{f(\bigvee\limits_i x^i) - f(x^i)} =0$, where $f(\bigvee\limits_i x^i) - f(x^i)$ exists because $f(\bigvee\limits_i x^i) \geq f(x^i)$ for all $i$ being $f$ monotone.
 We can therefore apply Proposition \ref{prop:infsup} to the directed net $(f(x^i))_i$ in $(0,\infty)$, obtaining that $\bigvee\limits_i f(x^i)$ exists in $(0,\infty)$ and it coincides with $f(\bigvee\limits_i x^i)$.
\end{proof}

\subsection{Proof of \autoref{thmTLSlocLip}}

\begin{figure}[h]
\begin{tikzpicture}[thick]
    \coordinate (u) at (-2,0);
    \coordinate (v) at (2,0);
    \coordinate (y) at (-1.2,0);
    \coordinate (z) at (0,0);

    \draw[name path={u--v}] (u) -- (v);
    \node [draw,name path=z] at (z) [circle through={(u)}] {};

    \path[name intersections={of=u--v and z,by={int_u,int_v}}]
      foreach \X in {u,v}{(int_\X) node[below left]{$\X$}};

    \fill   (int_u) circle[radius=2pt]  node [below left] {$u$};
    \fill   (int_v) circle[radius=2pt]  node [below left] {$v$};

    \fill   (z) circle[radius=0pt]  node [xshift=-1cm, yshift=2.2cm] {$\overline{B_{2\delta}(z)}$};

    \fill   (z) circle[radius=2pt]  node [above] {$z$};
    \fill   (y) circle[radius=2pt]  node [above] {$y$};

    \coordinate (x) at (-0.7,0.4);

    \draw   (x) circle[radius=1cm] node [above] {$x$};
    \fill   (x) circle[radius=2pt]  node [xshift=10mm, yshift=10mm] {$\overline{B_{\delta}(x)}$};

    \draw   (x) circle[radius=3cm] node {};
    \fill   (x) circle[radius=0pt]  node [xshift=-1.7cm, yshift=3cm] {$\overline{B_{3\delta}(x)}$};
\end{tikzpicture}
\caption{Drawing of the proof of \autoref{th:locLip}.}
\label{fig:proof_loc_lip}
\end{figure}

The main ingredient of the proof, that we mention in the proof sketch of \autoref{thmTLSlocLip}, is the following:

\begin{theorem}\label{th:locLip}
Let $f:V\subseteq (\BB R^X,\norm\cdot) \to (\BB R,\absv \cdot)$, with $V$ open and convex and $\norm\cdot$ any norm.
If $f$ is concave and locally bounded, then $f$ is locally Lipschitz.
Moreover, the Lipschitz constant of $f$ on $\overline{B_{\delta}(x)}$ can be chosen to be $\frac{1}{\delta}\max_{\overline{B_{3\delta}(x)}} \absv f$.
\end{theorem}
\begin{proof}
 Call $\overline{B_{\delta}(x)}:=B_1$, $\overline{B_{3\delta}(x)}:=B_3$.
 It suffices to show that for all $x\in V$, there is $\delta>0$ s.t.\ $B_3\subseteq \mathrm{interior}(V)$, $K:=\max_{B_3}  \absv f$ exists and $f$ is $(\frac{1}{\delta}\max_{B_3} \absv f)$-Lipschitz on $B_1$.
 A $\delta$ satisfying the first two conditions exists since $V$ is open and bcause $f$ is locally bounded and $B_3$ is compact.
 We will show that for all such $\delta$, the third condition already holds.

 For that, fix $y,z\in B_1$ and call $r:=\frac{d(y,z)}{2\delta}\in[0,1]$.
 We want to show that $\absv{f(y)-f(z)}\leq \frac{K}{\delta}d(y,z)=2Kr$.
 Wlog $y\neq z$, otherwise there is nothing to show.

 So $r\neq 0$ and we can consider $u:=\frac{1+r}{r}z-\frac{1}{r}y$, $v:=\frac{1}{r}y-\frac{r-1}{r}z$.
 We have $u,v\in \overline{B_{2\delta}(z)}=:B_2$.
 Indeed, $d(u,z)=\norm{u-z}=\norm{\frac{z}{r}+z-\frac{y}{r}-z}=\frac{\norm{z-y}}{r}=2\delta$ and similarly $d(v,z)=2\delta$.
 Geometrically, those are actually the intersections between $B_2$ and the line generated by $y$ and $z$, see \autoref{fig:proof_loc_lip}.
 Now we have the convex combinations $z=\frac{1}{1+r}y+\frac{r}{1+r}u$ and $y=(1-r)z+rv$, so the concavity of $f$ entails on one hand:
 $f(z)\geq \frac{1}{1+r}f(y)+\frac{r}{1+r}f(u)\geq \frac{f(y)}{1+r} - \frac{rK}{1+r}$, i.e.\ $f(y)-f(z)\leq r(K+f(z))\leq 2rK$, and on the other hand:
 $f(y)\geq (1-r)f(z)+rf(v)\geq f(z)-r(f(z)+K)$, i.e.\ $f(z)-f(y)\leq r(f(z)+K)\leq 2rK$.
 In the previous inequalities we have used that $f(u),f(v)\geq -K$.
 This follows because $u,v\in B_2\subseteq B_3$, as it can be immediately checked, thus $\absv{f(u)},\absv{f(v)}\leq K$.
 Putting the final inequalities together, we have $\absv{f(y)-f(z)}\leq 2rK$, i.e.\ the thesis.
\end{proof}

Therefore, since $(0,\infty)^X$ is open and convex in $(\BB R^X,\norm\cdot)$ and all tps are non-negative, we immediately have \autoref{thmTLSlocLip}.

\section{Proofs from \autoref{sec:TayLip}: Lipschitz meets Taylor}

\subsection{Proof of  Theorem~\ref{thm:modelsTaylor}}

Theorem~\ref{thm:modelsTaylor} states the validity of Taylor expansion in $\LREL_!$.
We must check for it the following equation, given $f\in \LREL_!(C, B^{A})$ and $g\in \LREL_!(C,A)$:
$$
\mathsf{ev}\circ \langle f,g\rangle= \inf_{n\in \BB N}
\left\{(( \cdots (\Lambda^{-}(f) \underbrace{\star g)\cdots )\star g}_{n\text{ times}})\circ \langle \mathrm{id}, \infty\rangle\right\}
$$
where:
\begin{enumerate}
\item $\mathsf{ev}\in \LREL_!(B^{A}+A, B)$ is the canonical \emph{evaluation} morphism;

\item $\Lambda^{-}(\_):= \mathsf{ev}\circ (\_\times \mathrm{id})$ is the \emph{uncurry} operator;

\item given $f\in \LREL_!(C+A,B)$ and $g\in \LREL_!(C,A)$, 
$f\star g\in \LREL_!(C+A,B)$ is the morphism obtained by differentiating $f$ in its second component and applying $g$ in that component, i.e.~
$$
f\star g =  D(f)\circ \langle \langle \infty, g\circ \pi_{1}\rangle, \mathrm{id}_{C+A}\rangle.
$$ 

\end{enumerate}

We do it in the following $4$ steps.

\begin{enumerate}

\item Let us compute the morphism $\mathsf{ev}$ explicitly: $\mathsf{ev}\in \Lawv^{ \C M_{\mathsf{fin}}(  ( \C M_{\mathsf{fin}}(A)\times B)       +  A    ) \times B  }$ is given by
$$\mathsf{ev}_{\mu,y}=
 \begin{cases}
 0 & \text{ if } \mu=[ \langle\rho, y\rangle]  \oplus \rho \\
 \infty & \text{ otherwise}
 \end{cases}
 $$
and observe that, given $f\in \LREL_!(C, B^{A})$ and $g\in \LREL_!(C,A)$, 
$$
\big(\mathsf{ev}\circ \langle f,g\rangle \big)_{\chi, y}= 
\inf\Big \{ 
\sum_{i=1}^{m}g_{\chi_{i},x_{i}}+
f_{\chi', \langle [x_{1},\dots, x_{m}],y\rangle}
\mid 
x_{1},\dots, x_{m}\in A,
\chi= \chi'+\sum_{i=1}^{m}\chi_{i}, 
\Big \}
$$

\item Let us compute the morphism $\Lambda^{-}$ explicitly: given $g\in \LREL_!(C, B^{A})$, 
$\Lambda^{-}(g)\in \LREL_!(C+A, B)$ is given by 
$$
\big(\Lambda^{-}(g)\big)_{\rho\oplus\mu,y}=g_{\rho, \langle \mu,y\rangle}
$$

\item Let us compute the morphism $\star$ explicitly: $f\star g$ is given by 
$$
(f\star g)_{\rho\oplus\mu,y}=
\inf\Big\{
g_{\rho',x}+
f_{\rho''\oplus(\mu+x)}
\mid
x\in A,
\rho= \rho'+\rho''
\Big\}
$$

\item We can now conclude: given the definition of $\mathsf{ev}\circ \langle f,g\rangle$, to check the Taylor equation it is enough to check that, for all $N\in \BB N$, 
$$
\left((( \cdots (\Lambda^{-}(f) \underbrace{\star g)\cdots )\star g}_{N\text{ times}})\circ \langle \mathrm{id}, \infty\rangle\right)_{\chi,y}=
\inf\Big \{ 
\sum_{i=1}^{N}g_{\chi_{i},x_{i}}+
f_{\chi', \langle [x_{1},\dots, x_{N}],y\rangle}
\mid 
\begin{matrix}
x_{1},\dots, x_{N}\in A,\\
\chi= \chi'+\sum_{i=1}^{N}\chi_{i}
\end{matrix}
\Big \}
$$

Let us show, by induction on $N$, the following equality, from which the desired equality easily descends:
$$
\big(( \cdots (\Lambda^{-}(f) \underbrace{\star g)\cdots )\star g}_{N\text{ times}}\big)_{\chi\oplus\mu,y}=
\inf\Big \{ 
\sum_{i=1}^{N}g_{\chi_{i},x_{i}}+
f_{\chi', \langle\mu+ [x_{1},\dots, x_{N}],y\rangle}
\mid 
\begin{matrix}
x_{1},\dots, x_{N}\in A,\\
\chi= \chi'+\sum_{i=1}^{N}\chi_{i}
\end{matrix}
\Big \}
$$
\begin{itemize}

\item if $N=0$, the right-hand term reduces to 
$f_{\chi, \langle \mu, y\rangle}=(\Lambda^{-}(f))_{\chi\oplus\mu,y}$;

\item otherwise, let $F=(( \cdots (\Lambda^{-}(f) \underbrace{\star g)\cdots )\star g}_{N-1\text{ times}})$, so that by I.H.~we have
$$
F_{\chi\oplus\mu,y}=
\inf\Big \{ 
\sum_{i=1}^{N-1}g_{\chi_{i},x_{i}}+
f_{\chi', \langle\mu+ [x_{1},\dots, x_{N-1}],y\rangle}
\mid 
\begin{matrix}
x_{1},\dots, x_{N-1}\in A,\\
\chi= \chi'+\sum_{i=1}^{N-1}\chi_{i}
\end{matrix}
\Big \}
$$
Then we have
{\small
\begin{align*}
\big( F\star g\big)_{\chi\oplus\mu,y}
&=
\inf \left \{
g_{\chi',x}+F_{\chi''\oplus(\mu+x)}
\mid
x\in A, \chi=\chi'+\chi''
\right\}\\
&=
\inf\left \{ 
g_{\chi',x}+
\inf\left\{
\sum_{i=1}^{N-1}g_{\chi_{i},x_{i}}+
f_{\chi'', \langle\mu+ [x_{1},\dots, x_{N-1}]+x,y\rangle}
\mid 
\begin{matrix}
x_{1},\dots, x_{N-1}\in A,\\
\chi^{*}= \chi''+\sum_{i=1}^{N-1}\chi_{i}
\end{matrix}
\right\}
\ 
\Big\vert \ 
\begin{matrix}
x\in A,\\
\chi=\chi'+\chi^{*}
\end{matrix}
\right \}\\
&=
\inf\Big \{ 
g_{\chi',x}+
\sum_{i=1}^{N-1}g_{\chi_{i},x_{i}}+
f_{\chi'', \langle\mu+ [x_{1},\dots, x_{N-1}]+x,y\rangle}
\mid 
\begin{matrix}
x,x_{1},\dots, x_{N-1}\in A,\\
\chi= \chi'+\chi''+\sum_{i=1}^{N-1}\chi_{i}
\end{matrix}
\Big \}\\
&=
\inf\Big \{ 
\sum_{i=1}^{N}g_{\chi_{i},x_{i}}+
f_{\chi', \langle\mu+ [x_{1},\dots, x_{N}],y\rangle}
\mid 
\begin{matrix}
x_{1},\dots, x_{N}\in A,\\
\chi= \chi'+\sum_{i=1}^{N}\chi_{i}
\end{matrix}
\Big \}.
\end{align*}
}
\end{itemize}
\end{enumerate}

\section{Proofs of \autoref{sec:GMS}: $\Lawv$-modules and Generalized Metric Spaces}

\subsection{Complete $\Lawv$-categories and their $\Lawv$-module structure}

In this subsection we quickly recall the notion of complete $\Lawv$-category and its associated $\Lawv$-module structure.

Functors of shape $\Phi: X\times Y^{\op}\to \Lawv$ are called \emph{distributors} and usually noted $\Phi: Y \pfun X$.

\begin{definition}[weighted colimits]
Let $X,Y,Z$ be $\Lawv$-categories,
$\Phi: Z\pfun Y$ be a distributor and  $f:Y\to X$ be a functor.
A functor $g:Z\to X$ is the \emph{$\Phi$-weighted colimit of $f$ over $X$}, noted $\colim(\Phi,f)$, if for all $z\in Z$ and $x\in X$
\begin{align}
X(g(z),x)= \sup_{y\in Y}\left\{X(f(y),x)\menus \Phi(y,z)\right\}
\end{align} 
A functor $f:X\to Y$ is \emph{continuous} if it commutes with all existing weighted colimits in $X$, i.e.~$f(\colim(\Phi,g))=\colim(\Phi,f\circ g)$. A $\Lawv$-enriched category 
$X$ is said \emph{categorically complete} (or just \emph{complete}) if all weighted colimits over $X$ exist. 
\end{definition}


%

An important example of colimit is the following:\begin{definition}[tensors]
Let $X$ be a $\Lawv$-category, $x\in X$ and $\epsilon \in \Lawv$. The \emph{tensor of $x$ and $\epsilon$}, if it exists, is the colimit $\epsilon \otimes x:= \colim( [\epsilon],\Delta x)$, where
$[\epsilon]: \{\star\}\pfun \{\star\}$ is the constantly $\epsilon$ distributor
and $\Delta x:\{\star\}\to X$ is the constant functor. 
\end{definition}

The $\Lawv$-module structure of a complete $\Lawv$-category has order given by $x\preceq_{X}y $ iff $X(y,x)=0$, and 
action given by tensors $\epsilon \otimes x$.

To conclude our correspondence between $\Lawv$-modules and complete $\Lawv$-categories, it remains to observe that the 
two constructions leading from one structure to the other are one the inverse of the other: for any $\Lawv$-module $(M,\preceq,\star)$,
$x\preceq_{M}y$ iff $M(y,x)=0$ iff $x\preceq y=0\star y$, and, from  
$M(\epsilon \star x, y)= M(x,y)\dotdiv \epsilon$, we deduce $\epsilon\otimes x=\epsilon \star x$. 
Conversely, 
for any complete $\Lawv$-category $X$ and $x,y\in X$, one can check that 
$X(y,x)=\inf\{ \epsilon \mid X(\epsilon\otimes y,x )=0\}$.

%
%
%
%

%
%


\subsection{Exponential and Differential Structure of $\Mod\simeq\GMet$}

In this subsection we show that the category $\Mod\simeq \GMet$ can be endowed with an exponential modality $!$ so that that the coKleisli category $\Mod_{!}$ is a model of the differential $\lambda$-calculus extending the category $\LREL_{!}$. 

First, we need to define a Lafont exponential $!$ over $\Mod$.
Since $\Mod $ is a SMCC with biproducts, where the latter commute with tensors
(see e.g.~\cite[Theorem 4.7.11]{Russo2007}), we can apply a well-known recipe from \cite{Mellies2018, Manzo2013}, which yields $!$ as the \emph{free exponential modality} (i.e.~such that $!X$ can be given the structure of the cofree commutative comonoid over $X$).

First, we define the symmetric algebra $!_{n}M:=\Sym_{n}(M)$ as the equalizer of all permutative actions on $n$-tensors $M\otimes \dots \otimes M$.
Notice that each element of $!_{n}M$ can be described as a join of ``multisets''
$[x_{1},\dots, x_{n}]$, where the latter is the equivalence class of the tensor
$x_{1}\otimes \dots \otimes x_{n}\in M^{\otimes_{n}}$ under the permutative actions.
Moreover, the $\Lawv$-module $!_{n}M$ is a complete $\Lawv$-category with distance function defined on basic ``multisets'' as follows:
\begin{align}
(!_{n}M)(\alpha,\beta)=
\sup_{\sigma\in \F S_{n}}\inf_{\tau\in \F S_{n}}\sum_{i=1}^{n}
X(x_{\sigma(i)},y_{\tau(i)})
\end{align}
where $\alpha=[x_{1},\dots, x_{n}]$ and $\beta= [y_{1},\dots, y_{n}]$, and extended to arbitrary elements $\alpha=\bigvee_{i}\alpha_{i}$ and $\beta=\bigvee_{j}\beta_{j}$
by $(!_{n}M)(\alpha,\beta)=\sup_{i}\inf_{j}(!_{n}M)(\alpha_{i},\beta_{j})$.

Finally, we define $!M$ as the infinite biproduct $\prod_{n}!_{n}M$, yielding the cofree commutative comonoid over $M$ (cf.~\cite[Proposition 1]{Mellies2018}).

The construction of $!$ for $\Mod$ generalizes the one for $\LREL$:
\begin{proposition}\label{prop:pinv3}
$!(\Lawv^{X})\simeq \Lawv^{\multiset(X)}$. In particular, $\Mod_{!}(\Lawv^{X},\Lawv^{Y})\simeq \LREL_{!}(X,Y)$.
\end{proposition}
\begin{proof}
Let us show that the morphism $h:\Lawv^{\C M_{n}(S)}\to \Lawv^{S\times \dots \times S}$ defined by 
$
h(f)(\langle s_{1},\dots, s_{n}\rangle)=h([s_{1},\dots, s_{n}])
$
is the equalizer of the diagram 
$
\begin{tikzcd}
\Lawv^{\C M_{n}(S)} \ar{r}{h} &\Lawv^{S\times \dots \times S}\ar{r}{[\sigma]} &
\Lawv^{S\times \dots \times S}
\end{tikzcd}
$, 
where $[\sigma](x)(\langle s_{1},\dots, s_{n}\rangle)=x(\langle x_{\sigma(1)},\dots, x_{\sigma(n)}\rangle)$, with $\sigma$ varying over $\F S_{n}$.

It is immediate that $h\circ [\sigma]=h\circ [\tau]$, for all $\sigma,\tau\in \F S_{n}$. Let now $k: C\to \Lawv^{S\times \dots \times S}$ satisfy $k\circ [\sigma]=k\circ [\tau]$: then for all $c\in C$, $k(c)(\langle s_{1},\dots, s_{n}\rangle)=k(c)(\langle s_{\sigma(1)},\dots, s_{\sigma(n)}\rangle)$, so $k(c)$ actually defines a unique element of $\Lawv^{\C M_{n}(S)}$, and thus $k$ splits in a unique way as $C \stackrel{k'}{\to} \Lawv^{\C M_{n}(S)} \stackrel{h}{\to}\Lawv^{S\times \dots \times S}$.

Now, observe that 
$\Lawv^{S\times \dots \times S}\simeq (\Lawv^{S})^{\otimes_{n}}$ (cf.~\cite[Corollary 4.7.12 ($iii$)]{Russo2007}), and then, since equalizers are unique up to a unique isomorphism, we obtain an isomorphism $\Lawv^{\C M_{n}(S)}\simeq !_{n}\Lawv^{S}$. From this we obtain the claim via $!(\Lawv^{X})=\prod_{n}!_{n}(\Lawv^{X})\simeq\prod_{n}\Lawv^{\mathcal M_{n}(X)}\simeq 
\Lawv^{\coprod_{n}\mathcal M_{n}(X)}\simeq
\Lawv^{\mathcal{M}_{\mathrm{fin}}(X)}$.
%
%
\end{proof}

At this point, \cite[Theorem 21]{LemayCALCO2021}, which states that an additive Lafont category with free exponential modality and finite biproducts is a \emph{differential category} \cite{Blute2006}, yields:
\begin{theorem}\label{thm:linearlemay}
$\Mod$ (equivalently $\GMet$) is a differential category. 
\end{theorem} 
Finally,  from Theorem \ref{thm:linearlemay} we can conclude that $\Mod_{!}$ can be endowed with a differential operator $E$ making it a CC$\partial$C \cite[Proposition 3.2.1]{Blute2009}, i.e.~Theorem \ref{thm:lemay} is proved.

To conclude, we make the definition of the differential operator $E$ of $\Mod_{!}$ explicit: for $f:!M\to N$, we let  
\begin{align}\label{eq:dermod}
Ef(\alpha)=
\bigvee\left\{
f(\beta\cup [x]) \ \Big \vert  \ 
\iota_{n}(\beta)\otimes \iota_{1}(x) \leq S(\alpha)
\right\}
\end{align}
where $\iota_{k}: M_{k}\to \prod_{i\in I}M_{i}$ is the injection morphism given by $\iota_{k}(x)( k)=x$ and $\iota_{k}(x)(i\neq k)=\infty$,
and $S: !(M\times N)\to !M\otimes !N$ is the Seely isomorphism \cite{Mellies2018}, and $E$ satisfies all required axioms.

One can easily check that, when $f\in \Mod_{!}(\Lawv^{X},\Lawv^{Y})\simeq \LREL_{!}(X,Y)$, its derivative $E f$ coincides with the derivative $D_{!}f$ defined for tps in Section \ref{section3}.

\end{document}